\documentclass[12pt]{amsart}
\usepackage{amsfonts}
\usepackage{enumerate}
\usepackage[letterpaper, left=2.5cm, right=2.5cm, top=2.5cm,
bottom=2.5cm,dvips]{geometry}
\usepackage{graphicx}
\usepackage{tabularx}
\usepackage{float}
\usepackage{amsmath}
\usepackage{amsthm}
\usepackage{xcolor}
\usepackage{times}
\usepackage[utf8]{inputenc}
\usepackage{caption}

\usepackage{verbatim}
\usepackage{xspace}
\usepackage{epsfig}
\usepackage{epstopdf}
\usepackage{subcaption}
\usepackage{url}
\usepackage[normalem]{ulem}
\usepackage{amssymb}
\usepackage{amsmath}

\usepackage{amsfonts}
\usepackage{amsthm}
\usepackage{tikz}
\usetikzlibrary{arrows.meta}
\usepackage{hyphenat}

\usepackage{latexsym}

\newtheorem{theorem}{Theorem}[section]

\newtheorem{corollary}[theorem]{Corollary}

\newtheorem{lemma}[theorem]{Lemma}

\newtheorem{proposition}[theorem]{Proposition}

\newcommand{\nop}[1]{}
\newcommand{\mms}{\mathit{mms}}
\newcommand{\alloc}{\mathit{allocate}}

\newcommand{\cA}{\mathcal{A}}
\newcommand{\cB}{\mathcal{B}}
\newcommand{\cC}{\mathcal{C}}
\newcommand{\cH}{\mathcal{H}}
\newcommand{\cZ}{\mathcal{Z}}
\newcommand{\cY}{\mathcal{Y}}
\newcommand{\cX}{\mathcal{X}}

\newcommand{\green}[1]{{\color{olive}{#1}}}
\newcommand{\red}[1]{{\color{red}{#1}}}

\begin{document}
\title[Maximin Share Allocations on Cycles]{Maximin Share Allocations  on Cycles\footnote{A preliminary version of the paper appeared in the proceedings of IJCAI 2018.}}

\author{Zbigniew Lonc}
\address{Warsaw University	of Technology, Poland}
\email{z.lonc@mini.pw.edu.pl}
\author{Miroslaw Truszczynski}
\address{University of Kentucky, USA}
\email{mirek@cs.uky.edu}

\begin{abstract}
The problem of fair division of indivisible goods is a fundamental problem 
of social choice. Recently, the problem was extended to the case when
goods form a graph and the goal is to allocate goods to agents so that each
agent's bundle forms a connected subgraph. For the maximin share fairness
criterion researchers proved that if goods form a tree, allocations
offering each agent a bundle of at least her maximin share value always
exist. Moreover, they can be found in polynomial time. We consider here the
problem of maximin share allocations of goods on a cycle. Despite the
simplicity of the graph, the problem turns out to be significantly harder than
its tree version. We present cases when maximin share allocations of goods
on cycles exist and provide results on allocations guaranteeing each agent a
certain portion of her maximin share. We also study algorithms for computing
maximin share allocations of goods on cycles.
\end{abstract}

\maketitle

\section{Introduction}

Fair allocation of indivisible goods is a fundamental problem of social choice 
\cite{BramsT1996,BouveretCM15}. It assumes a set of elements, referred to as
\emph{goods}, and a collection of agents each with her own utility function on
the sets, or \emph{bundles}, of goods. The utility functions are commonly assumed to be
additive and so it is enough to specify their values on individual goods only.
The objective is to assign to agents disjoint subsets of goods in a way that
meets some fairness criteria. Among the most commonly studied ones are
\emph{proportionality} and \emph{envy-freeness}, adapted to the case of
indivisible goods from the problem of fair allocation of divisible goods
or \emph{cake-cutting} \cite{BramsT1996,Procaccia15}, as well as recently
proposed
\emph{maximin} and \emph{minimax} share \cite{Budish2011,BouveretL16}. For
each of the criteria, it is of interest to identify classes of instances when
fair allocations exist, to establish the complexity of deciding the existence
of fair allocations, and to design algorithms for computing them.

In this paper, we focus on the maximin share criterion \cite{Budish2011}.
It is a relaxation of envy-freeness and proportionality, and has a natural
interpretation in terms of some allocation protocols. In a few years since
it was first proposed it has already received substantial attention.
The criterion
turns out to be highly non-trivial. First, while it is easy to see that
envy-free and proportional allocations do not always exist, it is not at all
clear whether the same is true for the less restrictive maximin criterion.
We now know it is. But it took a few years before the first examples
showing that maximin share allocations are not guaranteed to exist were
found \cite{ProcacciaW14}. Moreover, they turned to be quite intricate.
Further, the complexity of deciding the existence of maximin share allocations
has not yet been determined \cite{BouveretL16}. To get around the difficulty
of constructing maximin allocations, researchers proposed to relax the maximin
share criterion by requiring that the value of each agent's share in an
allocation be at least equal to some positive fraction of the maximin share.
Procaccia and Wang \cite{ProcacciaW14} proved that an allocation guaranteeing
agents at least 2/3 of their maximin share always exists, and that it can be
found in polynomial time if the number of agents is fixed (not part of input).
Both the result and the algorithm are based on deep combinatorial insights.
Building on that work, Amanatidis, Markakis, Nikzad and
Saberi \cite{AmanatidisMNS15} proved that for every constant $\epsilon$,
$0<\epsilon<2/3$, there is a polynomial-time algorithm finding an allocation
guaranteeing to each agent at least $(2/3 -\epsilon)$ of their maximin share
(with the number of agents a part of the input). Ghodsi, Hajiaghayi,
Seddighin, Seddighin and Yami \cite{GhodsiHSSY18} proved that in the results
by Procaccia and Wang and by Amanatidis \textit{et al.}, the constant
2/3 can be replaced with 3/4.

We study maximin share allocations of indivisible goods in the setting proposed
by Bouveret, Cechl\'arov\'a, Elkind, Igarashi and Peters \cite{BouveretCEIP17}. In the original problem, there are no restrictions
on sets of goods that can be allocated to agents. This ignores important
practical constraints that may make some sets highly undesirable. For
instance, goods may be rooms and labs in a building to be allocated to
research groups \cite{BouveretCEIP17}, or plots of land to be consolidated
\cite{KingB1982}. In such cases, legal sets of goods that could be allocated
to an agent might be required to form connected subgraphs in some graph
describing the neighborhood relation among goods (offices spanning segments
of a hall, plots forming contiguous areas of land).

Bouveret \emph{et al.} \cite{BouveretCEIP17}
studied envy-free, proportional
and maximin share allocations for that setting obtaining several interesting
complexity and algorithmic results. Our paper extends their study for the
maximin share criterion. In a striking positive result,
Bouveret \emph{et al.} \cite{BouveretCEIP17}
proved that maximin share allocations of goods on
trees always exist and can be found in polynomial time. In our work we look
beyond trees and show that as soon as the underlying graph has a single cycle,
the picture becomes much more complicated. Our main contributions are as
follows.

\smallskip
\noindent
1. We show that for goods on a cycle the maximin share value for an agent
can be computed in polynomial time. In two cases, when $m\leq 2n$ and
when agents can be grouped into a fixed number of types (agents are of the same
type if they have the same utility function), this allows us to design
polynomial time algorithms for computing maximin share allocations of $m$
goods (on cycles) to $n$ agents or determining that such allocations do not
exist.

\smallskip
\noindent
2. We show that deciding the existence of maximin share allocations of goods
on an arbitrary graph is in the class $\Delta_2^P$. For complete graphs
(the setting equivalent to the original one) this result improves the bound
given by
Bouveret and Lema{\^{\i}}tre \cite{BouveretL16}. %Moreover,
We further
improve on this upper bound for cycles and more generally, unicyclic graphs
by showing that for such graphs the existence of a maximin share allocation
is in NP.

\smallskip
\noindent
3. We obtain approximation results on the existence of allocations of goods
on a cycle that \emph{guarantee} all agents a specified fraction of their
maximin share value. While it is easy to show that for any number of agents
there are allocations guaranteeing each agent at least $1/2$ their maximin
share, improving on the guarantee coefficient of $1/2$ is a non-trivial
problem. We show that it can be improved to $(\sqrt{5}-1)/2$ ($>0.618$).
Further improvements are possible if we limit the number of agents or the
number of types of agents. In particular, we show that for the three-agent
case, there are allocations guaranteeing each agent $5/6$ of their maximin
share; for an arbitrary number of agents of up to three types there
are allocations guaranteeing each agent $3/4$ of their maximin share; and
for any number of agents of $t\geq 4$ types there are allocations
guaranteeing each agent $t/(2t-2)$ of their maximin share. In each case,
these allocations can be found in polynomial time. Moreover, the constants 5/6
and 3/4 for the cases of three agents and any number of agents of up to three types,
respectively, are best possible, that is, they cannot be replaced with any
larger ones.

\section{Definitions and Basic Observations}

A \emph{utility} function on a set $V$ of \emph{goods} (items) is a function
assigning \emph{non-negative} reals (utilities) to goods in $V$.
%\red{such that at least one good is assigned a positive value (we comment
%on this assumption below).}
We extend utility functions to subsets of $V$ by assuming additivity.

Following Bouveret \emph{et al.} \cite{BouveretCEIP17}, we consider the
case when goods form nodes of a certain \emph{connected} graph $G=(V,E)$.
We adopt the assumption of the connectedness of the graph for the
entire
paper and do not mention it explicitly again. In particular, whenever we use
the term \emph{graph} we mean a \emph{connected graph}. We write $V(G)$ and
$E(G)$ for the sets of nodes and edges of $G$ and refer to $G$ as the graph of
goods. Given a graph $G$ of goods, we will often refer to utility functions on
(the set) $V(G)$ as utility functions on (the graph) $G$.

Let $V=V(G)$ be a set of goods. A \emph{$G$-bundle} is a subset of $V$
that induces in $G$ a connected subgraph. A \emph{$(G,n)$-split} is a
sequence $P_1,\ldots, P_n$ of pairwise disjoint $G$-bundles such that
\[
\bigcup_{i=1}^n P_i =V.%
\]
We use the term \emph{split} rather than partition as we allow bundles to
be empty. When $G$ or $n$ are clear from the context we drop them from the
notation and speak about bundles and splits (occasionally, $G$-splits and
$n$-splits).

Let $G$ be a graph of goods, $u$ a utility function on $G$, $q$ a
non-negative real, and $n$ a positive integer. We call a
split \emph{$q$-strong} if every
bundle in the split has value at least $q$ under $u$. The \emph{maximin share}
for $G$, $u$ and $n$, written $\mms^{(n)}(G,u)$, is defined by setting
\[
\mms^{(n)}(G,u) = \max\{q: \mbox{there is a $q$-strong $n$-split of $G$}\}.
\]
An equivalent definition is given by
\[
\mms^{(n)}(G,u) = \max_{P_1,\ldots,P_n} \min_{i=1,\ldots,n} u(P_i),
\]
where the maximum is taken over all $n$-splits $P_1,\ldots, P_n$ of $G$.
A split for which the maximum is attained is a \emph{maximin share split}
or an \emph{mms-split} for $n$ and $u$ or for $n$ and an agent with a utility
function $u$. We often leave $n$ and $G$ implicit when they are clearly
determined by the context. By the definition, for every bundle $P$ in an
mms-split for $n$ and $u$, $u(P) \geq \mms^{(n)}(G,u)$. Similarly as above,
when $G$ and $n$ are clear from the context, we leave them out and write
$\mms(u)$ for $\mms^{(n)}(G,u)$. When considering an agent $i$ with a
utility function $u_i$ we write $\mms(i)$ for $\mms(u_i)$.

Let $G$ be a graph of goods.
%and $N=\{1,\ldots,n\}$ a set of agents, each with
%her own utility function $u$ on $V(G)$.
A \emph{$(G,n)$-allocation} is an assignment of pairwise disjoint $G$-bundles
of goods to $n$ agents $1,2,\ldots,n$ so that all goods are assigned. Clearly,
$(G,n)$-allocations can be represented by $(G,n)$-splits, where we understand
that the $i$th bundle in the split is the bundle assigned to agent $i$.
Let $u_i$ be the utility function of agent $i$, $1\leq i\leq n$.
A $(G,n)$-allocation $P_1,\ldots, P_n$ is a \emph{maximin share} allocation,
or an \emph{mms-allocation}, if for every $i=1,\ldots,n$ we have $u_i(P_i)\geq
\mms^{(n)}(G,u_i)$.

%\red{It is clear that agents with all-0's utility functions can be ignored
%when considering the existence of mms-allocations or designing algorithms
%to compute them. Indeed, each such agent can be assigned an empty bundle of
%goods and the problem can be reduced in this way to the case when each agent
%has a positive utility for at least one of the goods. Therefore, from now on
%\emph{we will only consider non-all-0's utility functions}.}

%\begin{tikzpicture}
%\begin{scope}[every node/.style={circle,thick,draw}]
%\colorlet{Gray}{white!50!black!50}
    %\node (A) at (0,0) {$A$};
    %\node (B) at (0,3) {B};
    %\node (C) at (2.5,4) {C};
    %\node (D) at (2.5,1) {D};
    %\node (E) at (2.5,-3) {E};
    %\node (F) at (5,3) {F} ;
%\end{scope}
%
%\begin{scope}[>={Stealth[black]},
              %every node/.style={fill=white,circle},
              %every edge/.style={draw=red,very thick}]
    %\path [->] (A) edge (B);
    %\path [-] (B) edge node {$3$} (C);
    %\path [->] (A) edge node {$4$} (D);
    %\path [->] (D) edge node {$3$} (C);
    %\path [->] (A) edge node {$3$} (E);
    %\path [->] (D) edge node {$3$} (E);
    %\path [->] (D) edge node {$3$} (F);
    %\path [->] (C) edge node {$5$} (F);
    %\path [->] (E) edge node {$8$} (F);
    %\path [->] (B) edge[bend right=90] node {$1$} (E);
%\end{scope}
%\end{tikzpicture}

To illustrate the concepts introduced above, let us consider the graph
in Figure \ref{fig-ex}. We will call this graph $G$. The graph $G$ defines
the set of goods $\{v_1,\ldots, v_8\}$ and their adjacency structure.
The table in the figure shows three utility
functions $u_1$, $u_2$ and $u_3$ on the set of goods. The values of the
utility function $u_1$ are also shown by the corresponding goods in the
graph (it will facilitate our discussion below).

\begin{figure}[!ht]
\centering
\begin{subfigure}[b]{0.5\textwidth}
\centering
\begin{tikzpicture}
\begin{scope}[every node/.style={circle,thick,draw}]
    \node (A) at (0,1) {$v_1$};
    \node (B) at (1,1) {$v_2$};
    \node (C) at (2,1) {$v_3$};
    \node (D) at (3,1) {$v_4$};
    \node (A1) at (0,0) {$v_8$};
    \node (B1) at (1,0) {$v_7$};
    \node (C1) at (2,0) {$v_6$};
    \node (D1) at (3,0) {$v_5$};
\end{scope}

\begin{scope}%[every node/.style={circle,thick,draw}]
    \node at (0.3,1.5) {$3$};
    \node at (1.3,1.5) {$1$};
    \node at (2.3,1.5) {$1$};
    \node at (3.3,1.5) {$4$};
    \node at (0.3,-0.5) {$4$};
    \node at (1.3,-0.5) {$1$};
    \node at (2.3,-0.5) {$1$};
    \node at (3.3,-0.5) {$3$};
\end{scope}

\begin{scope}[>={Stealth[black]},
              every node/.style={fill=Gray,circle},
              every edge/.style={draw=black,thick}]
              %every edge/.style={draw=red,very thick}]
    \path [-] (A) edge (B);
    \path [-] (B) edge (C);
    \path [-] (C) edge (D);
    \path [-] (D) edge (D1);
    \path [-] (D1) edge (C1);
    \path [-] (C1) edge (B1);
    \path [-] (B1) edge (A1);
    \path [-] (A1) edge (A);
    \path [-] (B1) edge (B);
    \path [-] (C1) edge (C);
\end{scope}
\end{tikzpicture}
\end{subfigure}
\ \\
\ \\
\ \\
\begin{subfigure}[b]{0.5\textwidth}
\centering
\begin{tabular}{|l|c|c|c|c|c|c|c|c|}
  \hline
 & $v_1$ & $v_2$ & $v_3$ & $v_4$ & $v_5$ & $v_6$ & $v_7$ & $v_8$ \\
  \hline
$u_1$ & 3 & 1 & 1 & 4 & 3 & 1 & 1 & 4 \\
  \hline
$u_2$ & 2 & 2 & 0 & 3 & 1 & 3 & 1 & 3 \\
  \hline
$u_3$ & 1 & 3 & 2 & 3 & 0 & 3 & 2 & 3 \\
  \hline
\end{tabular}
\ \\
\ \\
\end{subfigure}
\qquad\
\caption{A graph of goods and three utility functions.}
\label{fig-ex}
\end{figure}
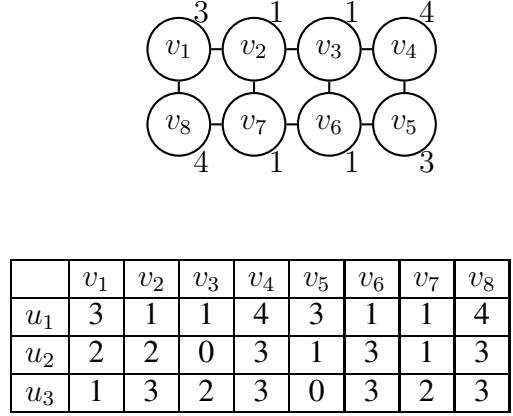

In our example, the set $\{v_1,v_2,v_6,v_7\}$ of goods is a
$G$-bundle (or, simply a bundle, as $G$ is clear). Indeed, the subgraph
of $G$ induced by $\{v_1,v_2,v_6,v_7\}$ is connected. On the other hand,
the set $\{v_1,v_2,v_5,v_6\}$ of goods is not a bundle as the corresponding
induced graph is not connected. Further, the sequence $\{v_1,v_7,v_8\},
\{v_2,v_3\},\{v_4,v_5,v_6\}$ is a $(G,3)$-split (or simply a 3-split).
Indeed, all sets in the split are bundles. It is also easy to see that
the sequence
$\{v_1,v_7,v_8\},\{v_2,v_4\},\{v_3,v_5,v_6\}$ is not a split as the set
$\{v_2,v_4\}$ is not a bundle.

Let us now consider the utility function $u_1$. The total utility
of all goods under $u_1$ is 18. Further, it is easy to see that 3-splits
into bundles of value 6 (under $u_1$) do not exist. To this end, we
note that $v_1$ and $v_8$ cannot be in the same bundle valued at 6 as their
total value is 7. Thus, the bundle containing $v_8$ must be either
$\{v_6,v_7,v_8\}$ or $\{v_2,v_7,v_8\}$. In each case the value of the bundle
containing $v_1$ is not 6. Thus, $\mms^{(3)}(G,u_1) < 6$. On the other hand,
the sequence $\{v_1,v_2,v_3\},\{v_4,v_5\},\{v_6,v_7,v_8\}$
is a 3-split with its bundles having values 5, 7 and 6, respectively. Thus,
$\mms^{(3)}(G,u_1) = 5$.

Reasoning in a similar way one can show that maximin share values for the
two other utility functions are also 5, that is, $\mms^{(3)}(G,u_2) =
\mms^{(3)}(G,u_3) = 5$. Let us now observe that the sequence $\{v_1,v_7,v_8\},
\{v_4,v_5,v_6\},\{v_2,v_3\}$ is a 3-split (indeed, its components are
bundles). Moreover, $u_1(\{v_1,v_7,v_8\}) = 8$, $u_2(\{v_4,v_5,v_6\})=7$
and $u_3(\{v_2,v_3\})= 5$. Thus, this 3-split defines an mms-allocation
of goods on $G$ to three agents with the utility functions $u_1$, $u_2$
and $u_3$. Later, we will also see examples when mms-allocations of goods on
graphs (specifically, cycles) do not exist.

We now are ready to present a few basic results on the problem of the
existence of mms-allocations. Two agents are of the \emph{same type} if
they have the same utility functions. The maximin share allocation problem
is easy when all but one agent are of the same type.

\begin{proposition}\label{prop3}
Let $G=(V,E)$ be a graph of goods. If we have $n$ agents and at most one
of them is of a different type than the others, then an mms-allocation exists.
\end{proposition}
\noindent

\begin{proof}
The result is obvious for $n=1$. Thus, we will assume $n\geq 2$. Let $u$ be
the utility function of agents $1,\ldots, n-1$ and $u'$ a utility function
of $n$. Let $\Pi$ be an mms-split for $u$, and $P$ a bundle of $\Pi$
most valuable under $u'$. Then, $u'(P)\geq \frac{\sum_{v\in V}u'(v)}{n} \geq
\mms(u')$. Assign $P$ to agent $n$. Next, allocate the remaining parts
of $\Pi$ to agents $1,\ldots,n-1$ in an arbitrary way. Since for each bundle
$Q$ in $\Pi$, $u(Q)\geq \mms(u)$, the resulting allocation is an mms-allocation.
\end{proof}

In particular, this result applies to the case of two agents.

\begin{corollary}\label{cor_2_agents}
Mms-allocations of goods on a graph to two agents always exist.
$\hfill\Box$
\end{corollary}

%\blue{
Let $G=(V,E)$ be a graph of goods. An agent with a utility function $u$
on $G$ (a utility function $u$ on $G$) is \emph{$n$-proportional}
if
\[
mms^{(n)}(G,u)=\frac{\sum_{v\in V} u(v)}{n}.
\]
We often omit $n$ from the notation if it is clear from the context. A set
$\{1,\ldots,n\}$ of agents (a set $\{u_1,\ldots,u_n\}$ of utility functions)
is \emph{proportional} if every agent $i$ (utility function $u_i$) is
$n$-proportional.

%Let $G=(V,E)$ be a graph of goods, and $u_1,\ldots,u_n$ be the utility
%functions of agents $1,\ldots, n$. An agent $i$ is \emph{proportional}
%if
%\[
%mms(i)=\frac{\sum_{v\in V} u_i(v)}{n}.
%\]
%The set of agents $\{1,\ldots,
%n\}$ is proportional if every agent $i$ is proportional. We also use the
%term ``proportional''
%%when speaking about the corresponding
%for utility functions and sets of utility functions.
%\blue{
An agent (a utility function) is \emph{$n$-regular} if it is
$n$-proportional and its maximin share (with respect to $n$) is 1.
%We note that the maximin share of an $n$-regular agent is 0 if and only if
%its utility function is an all-0 function. Finally,
A collection of $n$ agents (utility functions) is \emph{regular} if every
agent in the collection is $n$-regular. It is clear that for every agent in
an $n$-element regular collection of agents, the total value of all goods
for that agent is $n$. %}

%\red{A collection of $n$ agents (utility functions) is \emph{regular}
%if every agent in the collection is proportional and the total utility of
%all goods under every function equals $n$ (the number
%of agents). It is clear that for every regular collection of agents,
%the maximin share values for each agent are equal to 1.}

Let $c$ be a positive real. A bundle $P\subseteq V$ is
\emph{$c$-sufficient} for an agent $i$ if $u_i(P)\geq c\cdot \mms(i)$. An
allocation $\Pi=P_1,\ldots,P_n$ is \emph{$c$-sufficient} if for every
$i=1,\ldots,n$, $P_i$ is $c$-sufficient for $i$. Clearly, a $1$-sufficient
allocation is an mms-allocation.
The next result shows that when studying the existence of $c$-sufficient
allocations (and so, in particular, mms-allocations)
one can restrict considerations to the case when all agents are regular.
%it is enough to focus on situations when the collection of agents is
%proportional.
We use it to prove Theorem \ref{th1} later on in this section but its full
power becomes clear in Section \ref{approx}.

\begin{proposition}\label{prop2a}
Let $G$ be a graph of goods and $c$ a positive real such that for every
regular collection of $n$ agents (for every regular collection
of $n$ agents of at most $t$-types) a $c$-sufficient allocation exists.
Then a $c$-sufficient allocation exists for every collection of arbitrary
$n$ agents (for every collection of arbitrary $n$ agents of at
most $t$ types).
\end{proposition}
\begin{proof}
Let $N=\{1,\ldots,n\}$ be the set of agents and $\{u_1,\ldots,u_n\}$
the set of their utilities.
%We start by computing $\mms^{(n)}(u_i)$ for every $i\in N$.
If for every
$i\in N$, $\mms^{(n)}(u_i) = 0$, then the assertion holds. Any split of
$G$ into $n$ bundles is a $c$-sufficient allocation.

Thus, let us assume that for at least one agent, say $k$, $\mms^{(n)}(u_k)
>0$. For each agent $i\in N$ such that $\mms^{(n)}(u_i)=0$, we change her
utility function to $u_k$. We denote the new set of utilities by $v''_1,
\ldots,v''_n$. It is clear that if $P_1,\ldots,P_n$ is a $c$-sufficient
allocation with respect to the utility functions $v_i''$, then it is also a
$c$-sufficient allocation with respect to the original utilities $u_i$.
It is then enough to prove that a $c$-sufficient allocation with respect
to the functions $v_i''$ exists.

For every agent $i\in N$, we denote by $\Pi_i$ any mms-split into $n$
bundles with respect to the utility function $v''_i$. For each bundle of
$\Pi_i$ valued
more than $\mms^{(n)}(v''_i)$ we decrease the value of some elements in that
bundle to bring its value down to $\mms^{(n)}(v''_i)$. In this way, we produce a
new utility function, say $v_i'$. Since for every item $x$, $v_i'(x) \leq
v''_i(x)$, $\mms^{(n)}(v'_i) \leq \mms^{(n)}(v''_i)$. On the other hand, by
the construction, for every bundle $P\in\Pi_i$, $v'_i(P)=\mms^{(n)}(v''_i)$.
Thus, each utility function $v'_i$ is $n$-proportional and $\mms^{(n)}(v'_i)
= \mms^{(n)}(v''_i)>0$.

Next, for each agent $i\in N$ we set $k_i=n/(\sum_{y\in V(G)}v'_i(y))$
and, for every $x\in V(G)$, we define
\[
v_i(x)= k_i v'_i(x),
\]
Clearly, each utility function $v_i$, $i\in N$, is $n$-regular. Moreover,
it follows directly from the construction that if the original utility
functions are of at most $t$ types, then the resulting utility functions are
of at most $t$ types, too.

Let $\Pi=P_1,\ldots, P_n$ be a $c$-sufficient allocation for $v_1,
\ldots,v_n$ (its existence is guaranteed by the assumption). Therefore,
for every agent $i\in N$, $v_i(P_i) \geq c\cdot\mms^{(n)}(v_i)$.
Since $\mms^{(n)}(v_i)= k_i\mms^{(n)}(v'_i)$ and $v_i(P_i)=k_iv'_i(P_i)$,
it follows that $v'_i(P_i) \geq c\cdot\mms^{(n)}(v'_i)$. Next, we recall that
for every $x\in V(G)$, $v''_i(x)\geq v'_i(x)$. Consequently, for every bundle
$P$, $v''_i(P)\geq v'_i(P)$. Since for every $i\in N$, $\mms^{(n)}(v'_i) =
\mms^{(n)}(v''_i)$, it follows that
\[
v''_i(P_i) \geq v'_i(P_i) \geq c\cdot\mms^{(n)}(v'_i) = c\cdot\mms^{(n)}(v''_i).
\]
Thus, $\Pi$ is a $c$-sufficient allocation of goods to agents in $N$ with
respect to the utility functions $v''_i$ and the result follows.
\end{proof}

Corollary \ref{cor_2_agents} states that for two agents and for any connected
graph of goods an mms-allocation exists. On the other hand,
Bouveret \emph{et al.} \cite{BouveretCEIP17}
gave an example of nonexistence of an mms-allocation
for a cycle and four agents. In fact, as we show in Figure \ref{fig_no_mms_for_3}, even for three agents it may be that mms-allocations of goods on a cycle do
not exist.
In that figure, $v_1,\ldots,v_9$ denote consecutive nodes of a cycle and
the numbers in each row represent the utility functions. Observe that
the maximin shares for the agents 1 and 2 are $5$. For the agent 3, it is $6$.
Moreover, no two consecutive nodes have the total value satisfying any of
the agents. Therefore, if an mms-allocation existed, it would have to be a
split into three bundles of three consecutive nodes each. It is simple to
check that none of the three possible partitions of this type is an
mms-allocation.

\begin{figure}[ht]
{\small
\centerline{
\begin{tabular}{|l|c|c|c|c|c|c|c|c|c|}
  \hline
 & $v_1$ & $v_2$ & $v_3$ & $v_4$ & $v_5$ & $v_6$ & $v_7$ & $v_8$ & $v_9$\\
  \hline
agent 1 & 0 & 3 & 1 & 3 & 1 & 3 & 0 & 2 &  2 \\
  \hline
agent 2 & 2 & 2 & 0 & 3 & 1 & 3 & 1 & 3 & 0 \\
  \hline
agent 3 & 1 & 3 & 2 & 3 & 0 & 3 & 2 & 3 & 1 \\
  \hline
\end{tabular}}
\caption{An example of nonexistence of an mms-allocation for a cycle with $9$ nodes and $3$ agents.}
\label{fig_no_mms_for_3}
}
\end{figure}

On the other hand, for three agents and at most 8 goods on a cycle
mms-allocations always exist.

\begin{theorem}\label{th1}
For three agents and at most $8$ goods on a cycle, mms-allocations always exist.
\end{theorem}

\begin{proof}
Let $C$ be a cycle with at most $8$ nodes and let $N=\{1,2,3\}$ be the set
of agents. By Proposition \ref{prop2a}, we assume without loss of generality
that the agents are 3-regular. In particular, they are proportional
and $mms(i)=1$, for every $i\in N$.

Let us consider mms-splits for agents 1, 2 and 3, each into three
bundles. By regularity, all bundles in those splits are non-empty. Thus,
each split determines three edges that connect adjacent bundles of the split.
Since $C$ has 8 edges and there are three splits, there is an edge in $C$
that connects adjacent bundles in two different splits, say for agents 1 and 2.
It follows that for some bundles $A$ and $B$ of the mms-splits for agents
1 and 2, respectively, $A\subseteq B$ or $B\subseteq A$.

Let us assume without loss of generality that $A\subseteq B$. Clearly, the
elements of the set $B-A$ can be added to one of the two other parts of the
mms-split for the agent 2 to form a $2$-split of the path $C-A$, say $B',
B''$, with both parts of value at least $1$ for the agent $2$.

We now construct an mms-allocation as follows. If the set $A$ has value at
least 1 for the agent 3, then she receives it. Clearly, the value of $C-A$
for the agent 1 is equal to $2$. Thus, one of the parts $B',B''$ has value
at least 1 to agent 1. We allocate this part to agent 1. Finally, the agent
$2$ receives the part that remains.

If $A$ does not satisfy the agent $3$, then the agent $1$ gets this part.
The value of $C-A$ for the agent $3$, is larger than $2$. Thus, we allocate
the parts $B',B''$ to the agents $3$ and $2$ in the same way as to the
agents $1$ and $2$ in the previous case.
\end{proof}

We conclude this section with a comment on the notation. Formally, whenever
we talk about allocations of goods on a subgraph of a graph we consider, we
should use restrictions of the utility functions $u_i$ to the set of nodes
of the subgraph. Continuing to refer to $u_i$'s is simpler and does not
introduce ambiguity. Therefore, we adopt this convention in the paper.

\section{Complexity and Algorithms}

We assume familiarity with basic concepts of computational complexity and, in
particular, with the classes NP, $\Delta_2^P$ and $\Sigma_2^P$. We refer to
the book by Papadimitriou \cite{papa1993} for details.
%To make our discussion formal
We start with %some
general comments on how
instances of fair division of indivisible goods on graphs are given.
First, we assume any standard representation of graphs. As our primary
objective is to understand when problems we consider have polynomial solutions
and not most efficient algorithms to solve them, the details of how graphs
are represented are not critical. Moreover, many of our results concern
particular graphs such as cycles, when explicit representations of nodes and
edges are not even needed. Finally, we assume that all utilities and other
parameters that may be part of the input are rational, the least restrictive
assumption for studies of algorithms and complexity.
\nop{Second, we assume that all utility functions
are represented as sequences of non-negative integers. This does not affect
the generality of our results. The case when utility functions have
non-negative rational values, which is the least restrictive assumption for
studies of algorithms and complexity, can be reduced to the integer one
by multiplying all utility values by the least common multiple of their
denominators. It is routine to show that this process can be implemented to
run in polynomial time in the size of the original instance (which accounts,
in particular, for the total size of binary representations of all integers
that appear as numerators or denominators of rational numbers representing
values of the utility functions).
}

We now formally define several problems related to maximin share allocations
of indivisible goods on graphs and derive bounds on their complexity.

\smallskip
\noindent
{\sc Mms-Values-G}: Given a graph $G$ on $m$ goods, a utility function $u$,
that is, a sequence $u=(u_1, \ldots, u_m)$ of non-negative rational numbers,
an integer $n>1$ and a rational number $k\geq 0$, decide whether
$\mms^{(n)}(G,u)\geq k$.

\smallskip
The problem was studied in the original case (with $G$ being a complete graph)
by Bouveret and Lema{\^{\i}}tre \cite{BouveretL16}. Their complexity result
and its proof extend to the general case. The hardness argument uses a
reduction from the well-known NP-complete problem {\sc Partition}
\cite{GareyJ79}:

\smallskip
\noindent
{\sc Partition}: Given a set $U$ of $m$ non-negative
integers, decide whether there is a set $X\subseteq U$ such that
$\sum_{u\in X} u = \sum_{u\in U\setminus X} u$.

\begin{proposition}
\label{prop-mms-val}
The problem {\sc Mms-Values-G} is NP-complete.
\end{proposition}

\begin{proof}
The problem is clearly in NP. Indeed, to solve the problem, we guess an
$n$-split and verify that each of its parts induces in $G$ a connected
subgraph that has value at least $k$ (under $u$). The problem is NP-hard
even in the case when $G$ is a complete graph, $n=2$ and the utility function
is integer valued. We can show this by a reduction from {\sc
Partition}. Indeed, an instance of {\sc Partition} consisting of a set of
non-negative integers $\{u_1,\ldots,u_m\}$ yields an
instance to {\sc Mms-Values-G} with the complete graph of goods, $n=2$,
and $k=\lceil(\sum_{i=1}^m u_i)/2\rceil$. It is easy to observe that
an instance of {\sc Partition} is a YES instance if and only if the
corresponding instance of {\sc Mms-Values-G} is a YES instance. Thus,
NP-hardness follows.
\end{proof}

Another related problem concerns the existence of an allocation meeting
given bounds on the values of its individual bundles.

\smallskip
\noindent
{\sc Alloc-G}: Given a graph $G$ on $m$ goods, $n$ utility functions
$u_1,\ldots,u_n$ (each $u_i$ is a sequence of length $m$ consisting of
non-negative rational numbers) and $n$ rational numbers
$q_1,\ldots, q_n$, decide whether an allocation $A_1,\ldots, A_n$ for
$G$ exists such that for every $i=1, \ldots,n$, $u_i(A_i)\geq q_i$.

\begin{proposition}
The problem {\sc Alloc-G} is NP-complete.
\end{proposition}

\begin{proof}
The membership in NP is evident. The problem is NP-hard even if $G$ is a
complete graph, $n=2$, $q_1=q_2$, and $u_1=u_2$. Indeed, it becomes then
the {\sc Mms-Values-G} problem with $G$ being a complete graph, $n=2$,
$k=q_1$ ($=q_2$), and $u=u_1$ ($=u_2$). Thus, the proof of hardness
for Proposition \ref{prop-mms-val} can be applied here, too.
\end{proof}

We use these observations to show an upper bound on the complexity of the
problem of the existence of an mms-allocation in the graph setting. We
formally define the problem as follows.

\smallskip
\noindent
{\sc Mms-Alloc-G}: Given a graph $G$ of goods and $n$ utility functions
$u_1,\ldots, u_n$ (each $u_i$ is a sequence of length $m$ of non-negative
rational numbers), decide whether an mms-allocation for $G$ and $u_1,
\ldots, u_n$ exists.

\begin{theorem}\label{delta}
The problem {\sc Mms-ALLOC-G} is in the class $\Delta_2^P$.
\end{theorem}

\begin{proof}
The key step in the proof is the observation that we can rescale each
utility function so that all its values are integers. To this end, given
a utility function $u_i$, $1\leq i\leq n$, we compute the product, say $k_i$,
of all denominators of the rational numbers used by $u_i$ as its values.
Next, we multiply all values of $u_i$ by $k_i$ to produce a new utility
function $u'_i$. By the construction, all values of $u'_i$ are indeed
integers.

Let us denote by $I$ the original input instance to the problem and by
$I'$ the instance obtained by rescaling the utility functions in the way
described above. We note that the products $k_i$ and the rescaled utility
functions $u'_i$ can be computed in time $O(nm^2M^2)$, where $M$ is the number
of digits in the binary representation of the largest integer appearing as
the numerator or the denominator in rational numbers specifying the original
utility functions $u_i$. Since the size, say $S$, of the input instance
satisfies $S=\Omega(m+n+M)$, $I'$ can be computed in polynomial time in $S$.

Finally, we note that $I$ is a YES instance of the problem
{\sc Mms-ALLOC-G} if and only if $I'$ is a YES instance of the problem.
In fact, an allocation $\Pi$ is an mms-allocation for $I$ if and only if
$\Pi$ is an mms-allocation for $I'$.

Let us now consider an auxiliary problem where, given a graph $G$ on $m$
nodes and a sequence  $s=s_1,\ldots,s_m$ of $m$ non-negative integers (that
is, a utility function), the goal is to compute $\mms^{(n)}(G,s)$. We will
design for this problem an algorithm with an NP-oracle. We will then show that
our algorithm runs in polynomial time in the size of the representation of $s$,
where we consider each call to an oracle to be a single step taking a constant
amount of time.

To this end, we observe that $\mms^{(n)}(G,s) \leq (\sum_{j=1}^m s_j)/n$.
Thus, we can compute $\mms^{(n)}(G,s)$ by binary search on the range
$[0..\lfloor (\sum_{j=1}^m s_j)/n\rfloor]$ of possible values for
$mms^{(n)}(G,s)$. Each step starts with a range, say $[p..r]$, that contains
$\mms^{(n)}(G,s)$. We narrow down this range by making a call
to an oracle for the problem {\sc Mms-Values-G}, which we proved to be
NP-complete above. The call to the oracle is made on the instance consisting
of $G$, $s$ and
%\sout{$k=\lceil (p+k)/2\rceil$}
%\red
{$k=\lceil (p+r)/2\rceil$}. Depending on the oracle output, it
results in a new smaller range --- either
%\sout{$[0..\lceil (p+k)/2\rceil-1]$ or $[\lceil (p+k)/2\rceil .. r]$}
%\red
{$[p..k-1]$ or $[k .. r]$}. The process stops when the range is reduced
to just one element and that element is returned as $\mms^{(n)}(G,s)$.

The number of range-narrowing steps is given by $\log_2(\lfloor (\sum_{j=1}^m
s_j)/n \rfloor +1)$. Clearly,
\[
\log_2(\lfloor (\sum_{j=1}^m s_j)/n \rfloor +1) \leq \log_2(\sum_{j=1}^m s_j +1)
\leq \sum_{j=1}^m \log_2(s_j+1).
\]
Since the number of bits needed to represent a non-negative integer $x$ is
given by $\max(1, \lceil\log_2(x+1)\rceil$, it follows that the size, say $S'$,
of the representation of a problem instance satisfies $S'\geq \sum_{j=1}^m
\max(1, \lceil\log_2(s_j+1)\rceil$.
Thus, the number of range-narrowing steps is bounded by $S'$.
Consequently, the oracle algorithm we described runs indeed in polynomial
time in $S'$.

It follows that the problem {\sc Mms-Alloc-G} can be solved for an instance
$I$ by a procedure consisting of these three steps:
\begin{enumerate}
\item Compute the instance $I'$.
\item Compute the values $q_i = mms^{(n)}(G,u'_i)$, where $u'_1,\ldots,u'_n$
are the rescaled utility functions computed in step (1). This is done
using the oracle algorithm described above. It can be applied as all
utilities in $I'$ are integer-valued.
\item Decide whether there is an allocation $A_1,\ldots, A_m$ for $G$
such that $u'_i(A_i) \geq q_i$, for every $i=1,\ldots, m$, by invoking once
an NP-oracle for the problem {\sc Alloc-G}.
\end{enumerate}
We recall that $I'$ can be computed in polynomial time in $S$ (the size of
the original instance $I$). Based on that and on our discussion of the
auxiliary problem, step (2) of our algorithm also runs in polynomial time
in $S$, counting each oracle call as taking a unit amount of time. Finally,
the last step takes a single call to an oracle. Thus, the entire algorithm
runs in polynomial time in $S$, where we count each oracle call as taking
the unit amount of time. Since the oracles
used by the algorithm are for NP-complete problems, the problem
{\sc Mms-Alloc-G} is in $\Delta_2^P$.
\end{proof}

The upper bound $\Delta_2^P$ established by Theorem \ref{delta} applies also
in the case when $G$ is assumed implicitly, for instance, when it is a path,
a cycle or a complete graph represented by its set of nodes (but not edges).
It is because the number of oracle calls is bounded by the size of the
representation of the utility functions only. In the case of complete graphs,
Theorem \ref{delta} yields an improvement
on the bound $\Sigma_2^P$ obtained by Bouveret and Lema{\^{\i}}tre
\cite{BouveretL16}. We do not know if the upper bound of
$\Delta_2^P$ can be improved in general and, in particular, whether it can be
improved for complete graphs. On the other hand, we do know
that it can be
improved for trees. Bouveret \emph{et al.}
\cite{BouveretCEIP17} proved the following two results.

\begin{theorem}[Bouveret \emph{et al.} \cite{BouveretCEIP17}]
\label{bouv-vals}
There is a polynomial time algorithm that computes $\mms^{(n)}(T,u)$ and
a corresponding mms-split given a tree of goods $T$, a non-negative rational
utility function $u$ on the nodes (goods) of $T$, and an ingeter $n\geq 1$.
\end{theorem}

\begin{theorem}[Bouveret \emph{et al.} \cite{BouveretCEIP17}]
\label{bouv-alloc}
For every tree $T$ of goods and every set of $n$ agents with non-negative
rational utility functions on the goods in $T$, an mms-allocation exists.
Moreover, there is a polynomial-time algorithm to find it.
\end{theorem}

The results we presented suggest a question of the relationship between
the complexity of the {\sc Mms-Alloc-G} problem and the properties of the
underlying graph, as it becomes more complex than trees. The first step
towards understanding how the complexity grows is to analyze the case of
cycles and unicyclic graphs.

%in the
%case of trees, the bound can be improved down to the class P. In fact, all
%three problems discussed here have polynomial time solutions for trees
%\cite{BouveretCEIP17}. Further, as we show below, for cycles the problem
%{\sc Mms-Alloc} is in NP (whether it is NP-complete is open). This points
%to an interesting problem of finding the boundaries of change in the complexity
%of the {\sc Mms-Alloc} problem as the underlying graph becomes more complex.

First, we show that as in the case of trees (cf. \cite{BouveretCEIP17}),
the maximin share values and mms-splits for the case when goods form a
cycle (or a unicyclic graph) can
be computed in polynomial time.

\begin{theorem}\label{values}
There is a polynomial time algorithm for computing $\mms^{(n)}(U,u)$
and a corresponding mms-split, where $U$ is a unicyclic graph and
$u$ is a rational-valued utility function.
\end{theorem}

\begin{proof}
Using the rescaling technique discussed in the proof of Theorem
\ref{delta}, we can reduce in polynomial time the general problem to
the problem when all utilities are integers. To this end, we compute
the product, say $k$, of all denominators of the rational numbers that
are values of goods under $u$ (if all values are integers, we assume the
denominator is 1; this applies, in particular, to the case when all utilities
are 0). We then define a new utility function by multiplying the values of
$u$ by $k$. Once the maximin share value for the rescaled utility function
is computed, it is then used as the numerator of the maximin share for $u$
and $k$ is used as the denominator. Thus, to complete the proof it suffices
to describe a method of computing $\mms^{(n)}(U,u)$ and an mms-split under
the assumption that all values of $u$ are integers.

Let $C$ be the unique cycle of $U$. Every $U$-split has a bundle that
contains $C$ or is a split of the graph $U-e$ for some edge $e\in C$.
Thus,
\[
\mms^{(n)}(U,u)= \max(\max_{e\in C} \mms^{(n)}(U-e,u), \mms_C),
\]
where $\mms_C$ stands for
\[
\max_{P_1,\ldots,P_n} \min_{i=1,\ldots,n} u(P_i),
\]
with the maximum taken over all splits $P_1,\ldots, P_n$ that have a bundle
containing $C$.
To compute $\mms_C$, we proceed as follows. We construct a tree $U_C$
by contracting $C$ in $U$ to a single ``supernode,'' say $c$. Thus, the
nodes of $U_C$ are all nodes of $U-C$ and the supernode $c$. We define a
utility function $u'$ on the nodes of $U_C$ by setting
\begin{equation*}
u'(x) =
\begin{cases}
u(x) &\text{if $x$ is a node of $U-C$ (not the supernode)}\\
\sum_{y\in C} u(y) &\text{if $x=c$ (is a super node).}
\end{cases}
\end{equation*}
Clearly, $\mms_C = \mms^{(n)}(U_C,u')$. We note that $U_C$ and $u'$ can
be computed in polynomial time. Moreover, $U_C$ and all graphs $U-e$, $e\in C$
are trees. We now apply the algorithm by Bouveret et al. \cite{BouveretCEIP17}
(cf. Theorem \ref{bouv-vals}) to compute $\mms_C$ and $\mms^{(n)}(U-e,u)$,
for all edges $e\in C$, as well as the corresponding mms-splits. This takes
polynomial time. We then select the largest value among them and return it
together with its mms-split (in the case, the largest value is $\mms_C$, in
the bundle containing the superode, we replace it with the nodes of $C$).
\end{proof}

It is now a matter of routine to show that the {\sc Mms-Alloc-G} problem
for cycles and, more generally, for unicyclic graphs is in NP.

\begin{corollary}\label{uni-NP}
The problem {\sc Mms-Alloc-G} for unicyclic graphs is in NP.
\end{corollary}

\begin{proof}
Let $U$ be a unicyclic graph and $u_1,\ldots, u_n$ rational-valued
utility functions defined on the nodes of $U$. The following non-deterministic
polynomial-time algorithm decides the problem {\sc Mms-Alloc-G}: (1) guess an
allocation $P_1,\ldots,P_n$, (2) compute the values $\mms^{(n)}(U,u_i)$,
$i=1,\ldots, n$, and (3) verify that $u_i(P_i) \geq \mms^{(n)}(U,u_i)$, $i=1,
\ldots,n$. By Theorem \ref{values}, the step (2) can be accomplished in
polynomial time. Thus, the entire algorithm runs in polynomial time.
\end{proof}

We do not have an argument for NP-hardness of the problem {\sc Mms-Alloc-G}.
Thus, we do not know whether the bound obtained in Corollary \ref{uni-NP} is
tight in general. We do know, however, that in some special cases it can be
improved. We will now discuss these results.

Our first group of results is concerned with the case when the number $m$
of goods is small with respect to the number $n$ of agents, specifically,
when $m\leq 2n$. Our results rely on the following simple observation.

\begin{proposition}\label{prop_single}
Let $C$ be a cycle, $u$ a utility function on $C$, and $n\geq2$ an integer.
For every node $x$ of $C$, $mms^{(n-1)}(C-x,u) \geq mms^{(n)}(C,u)$.
\end{proposition}

\begin{proof} Let $\Pi$ be an mms-split of $C$ for $u$ and $n$.
Let $P$ be the bundle in $\Pi$ containing $x$, and let $P'$ and $P''$
be the two bundles in $\Pi$ inducing in $C$ segments neighboring the one
induced by $P$, respectively preceding it and succeeding it when traversing
$C$ clockwise ($P'=P''$ if $n=2$).
We move all goods in $P$ that precede $x$ to $P'$ and those that succeed $x$
to $P''$. In this way each bundle still spans a connected segment in $C$.
Next, we remove the ``new'' $P$ (at this point, $P$ consists of $x$ only).
The result is a split $\Pi'$ of $C-x$ into $n-1$ bundles, in which every
bundle has value at least $\mms^{(n)}(C,u)$. Thus, $\mms^{(n-1)}(C-x,u)
\geq\mms^{(n)}(C,u)$.
\end{proof}

We now are ready to show that, when there are $m$ goods on a cycle and
$n$ agents, and $m<2n$, then mms-allocations are guaranteed to exist.
Moreover, we also show that this result is sharp --- having exactly $m=2n$
goods allows for situations when mms-allocations
do not exist.

\begin{theorem}\label{m<2n}
If $m < 2n$, then an mms-allocation of $m$ goods on a cycle $C$ to
$n$ agents exists and it can be computed in polynomial time.
\end{theorem}

\begin{proof}
Let us consider agents $1,\ldots, n$ with utility functions $u_1,\ldots, u_n$.
Let $\Pi$ be an mms-split for the agent $n$. Clearly, for every bundle
$P\in \Pi$, $u_n(P) \geq \mms^{(n)}(C,u_n)$. Since $m<2n$, at least one
bundle in $\Pi$ consists of only one element, say $x$. It follows that $u_n(x)
\geq \mms^{(n)}(C,u_n)$. Since $C-x$ is a path (and so, a tree), Theorem
\ref{bouv-alloc} implies that there is an allocation giving each agent $j$,
$1\leq j\leq n-1$, a bundle valued at least $\mms^{(n-1)}(C-x,u_j)$. By Proposition
\ref{prop_single}, each such bundle is valued at least $\mms^{(n)}(C,u_j)$.
Thus, this allocation extended by the bundle $\{x\}$, allocated to the agent
$n$, forms an mms-allocation of the goods on $C$ to $n$ agents.

For an algorithm, we (1) compute the maximin share $\mms^{(n)}(C,u_n)$;
(2) select an item $x$ such that $u_n(x) \geq \mms^{(n)}(C,u_n)$ (as
argued above, such an $x$ exists); (3) construct an mms-allocation of
goods on the path $C-x$ to agents $1,\ldots, n-1$ (possible as $C-x$ is
a tree); and (4) extend the allocation constructed in (3) by giving
$\{x\}$ to $n$.

By our argument above, the allocation constructed in step (4) is an
mms-allocation. Steps (1) and (3) can be accomplished in polynomial time by
Theorems \ref{values} and \ref{bouv-alloc}, respectively. The same obviously holds for steps (2) and (4). Thus, the algorithm we described runs in
polynomial time.
\end{proof}

This result is sharp. If $m=2n$ and $n>3$ then an mms-allocation of $m$
goods on a cycle to $n$ agents may not exist as shown by the example in
Figure \ref{fig4}.

\begin{figure*}[ht!]
{\small
\centerline{
\begin{tabular}{|l|c|c|c|c|c|c|c|c|c|c|c|c|}
  \hline
 & $v_1$ & $v_2$ & $v_3$ & $v_4$ & $v_5$ & $v_6$ & $\ldots$ & $v_{2n-1}$ & $v_{2n}$\\
  \hline
agents $1,\ldots,n-2$  & $n$ & $1$ & $n-1$ & $2$ & $n-2$ & $3$ & $\ldots$ & $1$ & $n$ \\
  \hline
agents $n-1,n$  & $n$ & $n$ & $1$ & $n-1$ & $2$ & $n-2$ & $\ldots$ & $n-1$ & $1$  \\
  \hline
\end{tabular}}
\caption{An example showing that the result in Theorem \ref{m<2n}
is sharp.}
\label{fig4}
}
\end{figure*}

Indeed, all maximin share values are $n+1$. Thus, every bundle in any
mms-allocation would have to contain at least two elements. Since $m=2n$,
every bundle would have to consist of exactly two elements. There are only
two splits of a cycle with $m=2n$ nodes into $n$ bundles of two consecutive
goods. None
of them is an mms-allocation. The assumption that $n > 3$ is essential. We
have seen earlier that if $n=2$ then mms-allocations exist for every number
$m$ of goods. Thus, they exist for $m=2n=4$. If $n=3$, mms-allocations of
$m$ goods on a cycle exist for every $m\leq 8$ (cf. Theorem \ref{th1}). In
particular, they exist if $m=2n=6$.

As our example shows, if $m=2n$, there are cases when mms-allocations do
not exist. However, whether an mms-allocation exists in the case when $m=2n$
can be decided in polynomial time and, if so, an mms-allocation can be computed
efficiently.

\begin{corollary}\label{m=2n}
There is a polynomial time algorithm deciding the existence of an
mms-allocation of $2n$ goods on a cycle to $n$ agents, and computing one,
if one exists.
\end{corollary}

\begin{proof} Let $C$ be the cycle and $u_1,\ldots,u_n$ the
rational-valued utility
functions of the agents $1,\ldots, n$. We first compute the values
$\mms^{(n)}(C,u_i)$, $i=1,\ldots, n$ (cf. Theorem \ref{values}). If
for some item $x$ and agent $i$, $u_i(x)\geq \mms^{(n)}(C,u_i)$, then
reasoning as in the proof of Theorem \ref{m<2n} one can show that an mms-allocation for
$C$ exists, and that it can be found in polynomial time.
Otherwise, if there is an mms-allocation of goods on $C$ to $n$ agents,
every bundle in this allocation consists of two consecutive goods. There
are only two candidates for such allocations and one can check whether any
of them is an mms-allocation in polynomial time (as the values
$\mms^{(n)}(C,u_i)$ are known).
\end{proof}

The next result of this section concerns the case of $n$ agents of $t$ types,
where $t$ is fixed and is not a part of the input. Before we present our result
and prove it, we discuss how inputs and outputs are represented.

As we already noted earlier, instances to an mms-allocation problem on a
cycle are specified by a non-negative integer $m$ representing the number
of goods, and $n$ $m$-element sequences of non-negative rational
numbers, each sequence representing a utility function. We note that $m$
implicitly defines the goods as $v_1,v_2,\ldots,v_m$, as well as the cycle
as having
these goods appear on it in this order. Further, we adopt a natural
convention that for every $i=1,
\ldots, m$, the $i$th element in every utility sequence provides the utility
value for the good $i$.
%Clearly, the size of the input Under this representation, say $S$,
%satisfies $S=\Omega(nm)$.
This representation consists of $nm+1$ rational numbers (one of them, $m$,
an integer).\footnote{We note that we
do not include the number of agents, $n$, in the input; if we need it, we
can compute it by counting in the input the sequences representing the
utility functions.}

In the case when agents can be grouped into $t$ types, where $t$ is fixed
and not part of input, the case we are now considering, input instances can
be represented more concisely by a non-negative integer $m$, $t$ sequences
of $m$ non-negative rational numbers ($t$ utility functions) and $t$
non-negative integers $s_1,\ldots, s_t$, where each $s_r$ represents the
number of agents of \emph{type $r$} (that is, having $u_r$ as their utility
function). Thus, an instance consists of $tm=\Theta(m)$ rational numbers
represented by pairs of integers, and $t+1 =\Theta(1)$ integers.
We will use $M$ to denote the length of the binary representation of the
largest of the integers appearing in this representation, and $S$ to denote
the total size of the binary representations of the integers in the instance.
In particular, we have that  $S=\Omega(m+M)$.

We observe that if $n>m$, then the maximin share for each agent is 0 and,
consequently, every allocation of $m$ goods to $n$ agents is an
mms-allocation. Since $t$ is fixed, computing $n=s_1+\ldots+s_t$ takes
time $O(M)$ and, consequently, also $O(S)$. Thus, this case can be recognized
in time $O(S)$ and
then handled directly. Namely, we return a sequence of $t$ 0's as the maximin
shares for agents of types $1,2\ldots, t$ and, if an mms-allocation is
required, we allocate all goods to agent 1 and empty bundles of goods to
all other agents. We do not generate any explicit representation for this
allocation. It is implicitly identified by the all-0's output of the
maximin shares. This is to avoid having to output explicitly an $n$-element
sequence of bundles, as $n$, the length of this sequence, may be exponential
in the size of input.
It follows that the key case is then the case $n\leq m$.

\begin{theorem}\label{two-types}
For every integer $t\geq 1$, there is a polynomial time algorithm deciding
existence of an mms-allocation of $m$ goods on a cycle to $n$ agents of $t$
types (and computing one, if one exists).
\end{theorem}
\begin{proof}
Let us consider an instance to the problem given by an integer $m$ (the
number of goods), $t$ sequences $u_1,\ldots, u_t$, each consisting of $m$
non-negative rational numbers ($t$ utility functions), and $t$
non-negative integers $s_1,\ldots, s_t$ (the numbers of agents of type
$1,\ldots, t$, respectively).
We start by computing $n=s_1+s_2+\ldots+s_t$, which can be accomplished
in $O(S)$ time, as $t$ is fixed (we recall that we write $S$ for the size
of the input instance). The case $n>m$ has been discussed above.
Thus, in what follows we assume that $n\leq m$. An important consequence of
this assumption is that for every $r=1,\ldots, t$, $s_r\leq m$.

%Second, it allows us to represent
%allocations explicitly, as sequences of $n$ bundles. Since goods are
%identified by integers $1,\ldots, m$, bundles are disjoint and empty bundles
%can be represented in space of size $\Theta(1)$, it follows that the size of
%(an explicit representation of an allocation is $O((m+n)M)=O(m M)$.
%Since $S$, the size of input, satisfies $S=\Omega(m+M)$, allocations can be
%represented by sequences of size polynomial in $S$.

To decide whether there is an mms-allocation and, if so, to compute it, we
compute the values $q_r=\mms^{(n)}(C,u_r)$, $r=1,\ldots, t$.
%where $n =s_1+\ldots+s_t$.
Since $t$ is fixed, Theorem \ref{values} implies that all these values
can be computed in time bounded by a polynomial in $|u|$, where $|u|$
represents the largest size of the representation of an input utility function
and, consequently, also by a polynomial in $S$.

We now observe that an allocation $\Pi$ of goods on $C$ to agents of $t$ types
given by $u_1,\ldots, u_t$ is an mms-allocation on $C$ if and only if for some
edge $e$ of $C$, $\Pi$ is an allocation of goods on the path $C-e$ to these
agents such that for every agent of type $r$, the bundle $P\in \Pi$ allocated
to that agent satisfies $u_r(P)\geq q_r$.

Thus, to prove the assertion, it is enough to show that the following
problem has a polynomial-time solution: Given a \emph{path} $F$ of $m$
goods $v_1,v_2, \ldots,v_m$, appearing in this order on $F$, $t$
rational-valued utility functions $u_1,\ldots, u_t$, $t$ non-negative
integers $s_1,\ldots, s_t$ such that $n=s_1+\ldots+s_t\leq m$, and $t$
non-negative integers $q_1, \ldots, q_t$, find an allocation $\Pi$ of goods
on $F$ to $s_1+\ldots+s_t$ agents such that for every agent of type $r$,
the bundle $P$ assigned to that agent satisfies $u_r(P)\geq q_r$.

%We note that we can represent each allocation explicitly as a sequence
%of $n$ bundles. Indeed, the bundles are disjoint, empty bundles can be
%represented in space $\Teta(1)$.  Indeed, by our assumption
%We first note that without loss of generality we may assume that for every
%$r=1,\ldots, t$, $q_t>0$. Indeed, removing from the original instance, say $S$,
%of our problem all $s_r$ agents of type $r$ whenever $q_r=0$ (eliminating from
%the specification of the input sequence all integers $s_r$ such that $q_r=0$,
%and all integers $q_r=0$) yields an instance, say $S'$, of our problem with all
%values $q_r$ positive. It is easy to see that $S$ has a solution if and only
%if $S'$ has a solution, and that any solution to $S'$ is also a solution
%to $S$ under the representation of allocations that we adopted. We note that
%the reduction runs in constant time (as $t$ is fixed).

%Thus, from now on we will assume that for every $r=1,\ldots, t$, $q_r>0$. If
%$n=s_1+\ldots+s_t > m$, then there is no solution (in every allocation,
%some agent gets an empty bundle). It remains then to handle the case when
%$n \leq m$. Importantly, when $n\leq m$, we also have $s_r\leq m$, for every
%$r=1,\ldots, t$.

To describe our method, we define $\cH$ to be the set of all sequences
$(h_1,\ldots, h_t)$ such that for every $r=1,\ldots,t$, $0\leq h_r \leq s_r$.
Clearly, the number of sequences in $\cH$ is bounded by $(s_1+1)\cdot\ldots
\cdot (s_t+1)\leq (m+1)^t=O(m^t)=O(S^t)$. If $H$ denotes a sequence
$(h_1,\ldots, h_t)\in \cH$ and $h_r> 0$, where $1\leq r\leq t$, then we write
$H_r^{-}$ for the sequence $(h_1,\ldots,h_{r-1},h_r-1,h_{r+1},\ldots,h_t)$
which, we note, also belongs to $\cH$.

For a sequence $(h_1,\ldots, h_t)\in \cH$, we define $T(h_1,\ldots, h_t)$
to be the smallest $j$ such that there is an allocation of goods $v_1,\ldots,
v_j$ to $h_1+\ldots+h_t$ agents, exactly $h_r$ of them of type $r$, $1\leq
r\leq t$, so that each agent of type $r$ obtains a bundle worth at least
$q_r$. If such an allocation does not exist, $T(h_1,\ldots,h_t)$ is set to
$\infty$.

Next, for every $j=1,\ldots, m$ and every $r=1,\ldots,t$, we define $k(j,r)$
to be the minimum $k\geq j$ such that $u_r(\{v_{j+1},\ldots, v_k\})\geq q_r$.
We set $k(j,r)=\infty$ if such a $k$ does not exist. Clearly, all values
$k(j,r)$ can be computed in time $O(tmM)=O(m M)=O(S^2)$.

Let $(h_1,\ldots,h_t)\in \cH$. We will show that $T(h_1,\ldots,h_t)$ can
be efficiently computed. Clearly, $T(0,\ldots, 0)=0$ (there are no agents
to get any bundles and so, even the empty path suffices). Thus, let us
consider $H=(h_1,\ldots,h_t)\in\cH$ and let us assume that $h_1+\ldots+h_t
>0$. Let us define $I = \{r: h_r>0,\ 1\leq r\leq t\}$. It is easy to see that
\begin{equation}
\label{eq:dp}
%T(H) = min_{r\in I}(T(H_r^{-}) + k(T(H_r^{-}),r)).
T(H) = min\{k(T(H_r^{-}),r)\colon r\in I\}.
\end{equation}
Assuming that all values $k(T(H_r^{-}),r)$, for $r\in I$, are known,
$T(H)$ can be computed in time $O(t)=O(1)$. It follows that with the initial
value of $T(0,\ldots, 0)=0$, considering sequences $H\in\cH$ in any order
consistent with non-decreasing sums of their elements and using the formula
(\ref{eq:dp}) to compute $T(H)$, one can compute $T(s_1,\ldots, s_r)$ in
polynomial time, in fact, in time $O(S^{\max(2,t)})$ (we recall that
all values $k(j,r)$ can be computed in time $O(S^2)$; also, the number of
entries in $T$ is $O(S^t)$ and each entry can be computed in $O(1)$ time).

If $T(s_1,\ldots, s_t) =\infty$, then there is no split $\Pi$ of $F$ such
that for every agent of type $r$, its bundle $P\in \Pi$ satisfies $u_r(P)
\geq q_r$. Otherwise,
%such a split exists (and can be constructed in polynomial time in the standard
%way for dynamic programming algorithms if in each application of the formula
%(\ref{eq:dp}) we record an index $r\in I$ that minimizes the expression
%$k(T(H_r^{-}),r)$).
%
%Since the number of elements is $\cH$ is bounded by $O(S^t)$, the entire
%table can be computed in polynomial time in $S$. Moreover, if
$T(s_1,\ldots, s_t)\neq \infty$ and a split solving the problem exists.
Moreover, it can be computed in
the standard way for dynamic programming algorithms. To this end, each time
the formula (\ref{eq:dp}) is applied we have to record an index $r\in I$ that
minimizes the expression $k(T(H_r^{-}),r)$). This information
allows us to construct the split in polynomial time in $S$.

%The number of entries $T(H)$ that we must compute in order to derive
%$T(s_1,\ldots,s_r)$ is given by $(s_1+1)(s_2+1)\ldots(s_t+1)$. We recall
%that for every $r$, $1\leq r\leq t$, we have $s_i\leq m$. Thus, the number
%of entires $T(H)$ we must compute is given by $\Theta(m^t)=O(S^t)$.

As all algorithmic tasks we presented can be accomplished in time bounded
by a polynomial in S, the assertion follows.
\end{proof}

%We we will now show that if the original instance is given in terms
%of rational utlities and a rational $c$ then this reduction can be computed
%in polynomial time. Indeed, we can rescale by the least common multiple
%of the denominators of all utility values. This yields an instance with
%all values being integer. Using Theorem 3.6, we reduce this case to the
%situation when all utility functions are proportional. Next, we rescale
%each utility function by dividing its values by (thetotal/n). All phases
%are polynomial in the size of the original representation.

We close this section with yet another corollary to Theorem \ref{values}.
It concerns a possibility of regularizing utility functions on unicyclic
graphs in polynomial time by converting a method used in the proof of
Proposition \ref{prop2a} into an algorithm. Let us
call a collection $\{u_1,\ldots,u_n\}$ \emph{trivial} if for every $i$,
$\mms^{(n)}(u_i)=0$. We will call all other collections non-trivial.

\begin{corollary}
\label{regular}
There is a polynomial time algorithm that, given a unicyclic graph $U$
of goods, and a non-trivial collection of rational-valued utility functions
$u_i$, $1\leq i\leq n$, on $U$, produces a rational-valued regular collection
of utility functions $u'_i$ on $U$ such that for every $c$, if a split $\Pi$
is a $c$-sufficient allocation with respect to $u'_i$'s
then it is a $c$-sufficient allocation with respect to $u_i$'s.
If the original utility functions are of at most $t$ types,
then the resulting utility functions are of at most $t$ types.
\end{corollary}
\begin{proof}
The algorithm follows the method we used in the proof of Proposition
\ref{prop2a}. It consists of the following key steps.
\begin{enumerate}
\item Rescale the utilities $u_i$ as described in the proof of Theorem
\ref{delta} to produce an
equivalent collection of integer-valued utilities $w_i$ (by equivalent
we mean determining the same splits as mms-allocations).
\item Compute the maximin shares $q_i=\mms^{(n)}(U,w_i)$, $1\leq i\leq n$.
\item Select one utility function $w_k$ such that $q_k>0$. For each
agent $i$ such that $q_i=0$, replace $w_i$ with $w_k$; denote the new
utility functions $v_i$, $1\leq i\leq n$.
\item For a good $x$ and a utility function $v_i$, use the binary search
method, similar to that used in the proof of Theorem \ref{delta}, to find
the smallest $a$ such
that $\mms^{(n)}(U,v_i^{x\rightarrow a}) = q_i$; replace $v_i$ with
$v_i^{x\rightarrow a}$; repeat for all goods as long as it is possible
to decrease a value of some utility function on some good.\footnote{Here,
$v_i^{x\rightarrow a}$ stands for the utility function obtained from $v_i$
by setting its value on $x$ to $a$ and keeping all other values as they are
in $v_i$.}
\item Rescale the computed utility functions so that the values of each of
them sum up to $n$. Call the resulting functions $u_i'$.
\end{enumerate}
The correctness of the method, both for the main statement and under the
restriction to utility functions of at most $t$ types, follows from the
argument used to prove
Proposition \ref{prop2a}. For the running time, we note that step (1)
runs in polynomial time (we argued this
in the proof of Theorem \ref{delta}). Step (2) runs in polynomial time (Theorem \ref{values}).
Step (3) runs in time $O(nm)$. Next, we note that the argument we used
in the proof of Theorem \ref{delta} to estimate the running time of the
``range narrowing'' binary search and the fact that the maximin shares can
be computed in polynomial time (Theorem \ref{values}) together imply that
step (4) runs in polynomial time. Step (5) consists of $O(nm)$ additions and
multiplications involving utility values produced in step (4) and the
integer $n$. Thus, it also runs in polynomial time in the size of the
original instance.
\end{proof}

\section{Approximate Maximin Share Allocations On a Cycle}
\label{approx}

We start with a simple observation that $\frac{1}{2}$-sufficient allocations always
exist and, moreover, can be found in polynomial time.

\begin{proposition}\label{one_half}
There is a polynomial time algorithm which, for any number of agents and
any number of goods on a cycle, constructs a $\frac{1}{2}$-sufficient
allocation.
\end{proposition}

\begin{proof} We remove any edge $e$ from the cycle $C$ and get a path,
say $P$. The maximin share for each agent on $P$ is at least half of the
original maximin share for $C$. Indeed, let us consider an arbitrary agent
$i$ and her mms-split of $C$. By removing $e$, no more than one part of this
split may break into two pieces. The value for the agent $i$ of one of these
pieces is at least $\frac{1}{2}mms(i)$. We adjoin the other piece, if it is
present, to its neighboring part in the mms-split of $C$. In this way, we
obtain a split of $P$. Clearly, in this split of $P$ every piece is worth
to $i$ at least $\frac{1}{2}mms(i)$, so the claim follows.

Thus, a $\frac{1}{2}$-sufficient allocation can be found by (1) removing
any edge from the cycle, (2) applying to the resulting path the algorithm
by Bouveret \emph{et al.} \cite{BouveretCEIP17} that constructs
in polynomial time an mms-allocation for trees (cf. Theorem \ref{bouv-alloc}).
\end{proof}

To find a better guarantee than $\frac{1}{2}$ turns out to be non-trivial.
We will show that it can be improved to $(\sqrt{5}-1)/2 \approx 0.61803...$,
and that further improvements are possible if we restrict the number of agent
types.

Let %\sout{$\frac{1}{2}<c\leq 1$}
$c$ be a positive real and let $N=\{1,\ldots,n\}$, where $n\geq 2$, be
a set of agents with arbitrary utility functions $u_1,\ldots, u_n$ on a set
of goods on a path, say $P$.
%\red{In this section we will use the
%term \emph{easy} when referring to an agent in $N$ (her utility function)
%with the maximin share 0.}
Figure \ref{fig2} shows an algorithm $\alloc$ that assigns to some
(possibly all) agents in $N$ bundles they value at $c$ or more.

\begin{figure}[ht]
\begin{tabbing}
\quad\=\quad\=\quad\=\quad\=\quad\=\quad\=\quad\=\quad\=\quad\=\quad\=
\quad\=\quad\=\\
$\alloc(N,P,Q,c)$\\
\>\>\% $P$ is a path; we fix its direction so that prefixes of $P$ are well
           defined\\
\>\>\% $N=\{1,\ldots,n\}$, $n\geq 2$, is a set of agents, each with
a utility function on $P$\\
\>\>\% $c$ is a positive real\\
\>\>\% $Q$ is a prefix of $P$ valued at least $c$ by at least one agent
in $N$\\
\>\ \\
1\>\>$S:=\{i\in N$: there is an $(n-1)$-split of $P-Q$ that is $c$-strong for $i\}$;\\
2\>\>$R:=N-S$;\\
3\>\>$j:=1$;\\
4\>\>{\bf while} $j\leq n$ and $P$ has value at least $c$ for some $i\in S\cup R$ {\bf do}\\
5\>\>\>$Q_j:=$ the shortest prefix of $P$ worth at least $c$ for some agent
               $i\in R\cup S$\\
6\>\>\>\textbf{if} $Q_j$ is worth at least $c$ to an agent $i\in R$
                            \textbf{then}\\
7\>\>\>\>assign $Q_j$ to $i$;\\
8\>\>\>\>$R:= R - \{i\}$\\
9\>\>\>\textbf{else}\\
10\>\>\>\>assign $Q_j$ to an agent $i\in S$ such that $Q_j$ is worth at least
    $c$ for $i$;\\
11\>\>\>\>$S:= S - \{i\}$;\\
12\>\>\>$P:=P-Q_j$;\\
13\>\>\>$j:=j+1$
%\red{14}\>\>\red{Assign an empty bundle to every agent $i\in R$ such that
%$\mms^{(n)}(i)=0$}\\
%\>\>\>\red{and remove these agents from $R$;}
\end{tabbing}
\caption{Algorithm $allocate$}
\label{fig2}
\vspace*{-0.15in}
\end{figure}

We will use this algorithm in our theoretical considerations on the
existence of $c$-sufficient allocations. Under the restriction to
rational-valued utilities and a rational $c$, the algorithm runs in
polynomial time. We will use it to argue the existence of polynomial-time
algorithms for constructing $c$-sufficient allocations.

First, let us discuss the algorithm $\alloc$ informally. The algorithm
starts by defining a set $S$ to consist of all agents $i$ in $N$ for whom
an $(n-1)$-split of $P-Q$ that is $c$-strong for $i$ can be found. The
algorithm then sets $R=N-S$.
We note that $N=S\cup R$, $S\cap R=\emptyset$ and that it may be that one of
the sets $S$ and $R$ is empty.

Next, the algorithm sets $j$ to 1 and proceeds to the loop (4--13). Each
time the body of the loop is executed, it assigns a bundle to an
``unassigned'' agent. Throughout the execution of the loop, $P$ denotes the
path consisting of unallocated goods, and $R\cup S$ contains all agents that
are as yet unassigned.
In each iteration $j$, the algorithm attempts to allocate a bundle
to an ``unassigned'' agent. At the start of that iteration $j-1$ agents have
received bundles selected as prefixes of the paths being considered in earlier
iterations. If $j>n$, then the loop terminates and all agents are assigned
bundles. If $j\leq n$, some agents are unassigned, $R\cup S$ contains all
unassigned agents and $P$ is a path on unallocated goods. If the value of
$P$ for each unassigned agent is less than $c$, no further assignments are
possible and the loop terminates. Otherwise $P$ has value at least $c$ for
at least one unassigned agent and an assignment can be made. The bundle for
the assignment, denoted by $Q_j$ is chosen as a \emph{shortest prefix} of
$P$ that has value at least $c$ for some unassigned agent. Selecting $Q_j$
as a prefix ensures that the remaining goods form a path. Selecting for $Q_j$
a shortest prefix that has value $c$ or more for some unassigned agent is
essential for a key property of the algorithm, which we will discuss later.
Once the bundle $Q_j$ is constructed, it is assigned. By construction,
there are unassigned agents that value $Q_j$ at $c$ or more. We select one
such agent, say $i$. We first check if such an agent $i$ can be found in
$R$ and if so, assign $Q_j$ to $i$.  Only if no agent in $R$ values $Q_j$
at $c$ or more, we select $i$ from $S$ (this selection is possible as at
least one unassigned agent values $Q_j$ at $c$ or more), and assign $Q_j$
to that $i$.
%\red{In the last step of the algorithm, every agent $i\in R$
%(hence, still unassigned) and such that $\mms^{(n)}(i)=0$ is allocated an
%empty bundle.}

The following proposition gives the key property of the algorithm $\alloc$.

\begin{proposition}\label{prop_S}
Let $P$ be a path of goods and $N=\{1,\ldots, n\}$, where $n\geq 2$, a set
of agents, each with a utility function on $P$, and $c$ a positive real.
Further, let $Q$ be a prefix of $P$ valued at least $c$ by at least one agent in $N$, and $S$ the set
computed by the algorithm $\alloc(N,P,Q,c)$ in line (1). When the algorithm
$\alloc(N,P,Q,c)$ terminates, $S=\emptyset$, that is, all agents included in
$S$ in line (1) are assigned bundles they value at $c$ or more.
\end{proposition}
\begin{proof}
We will denote by $P_0$ the original path $P$ and by $S_0$ the set $S$ as
computed in line (1). Let us assume that when the algorithm $\alloc(N,P_0,Q,c)$
terminates, $S\neq \emptyset$. Let $i\in S$. Since after line (1) the algorithm
never includes elements in $S$, it follows that $i\in S_0$.

Let us consider the value, say $j_t$ of the variable $j$ when the loop
(4--13) terminates. Clearly, the body of the loop was executed $j_t-1$
times and $j_t-1$ agents are assigned bundles $Q_1,\ldots, Q_{j_t-1}$.
By the conditions on the input parameters, the body of the loop executes
for $j=1$ and defines a bundle $Q_1$. This bundle is a prefix of $Q$ (because
of how we select prefixes $Q_j$).
Since $i\in S_0$, $P_0-Q$ has an $(n-1)$-split that is $c$-strong for $i$.
Since $Q_1$ is a prefix of $Q$, $P_0-Q_1$ has an $(n-1)$-split that is
$c$-strong for $i$. Let us denote this split by $D_2,\ldots, D_n$. Since
$i$ has not been assigned a bundle in iteration 1 (in any iteration, in fact)
and $D_2$ has value at least $c$ for $i$, at the beginning of iteration 2 we
have that $i\in S\cup R$ and the value of the path $P$ for $i$ is at least
$c$. Thus, the body of the loop executes for the second time. It follows that
$3\leq j_t$. Further, since $i$ is not assigned a bundle,
$j_t-1\leq n-1$.  Thus, $3\leq j_t\leq n$.

Let us assume that $(Q_1\cup \ldots\cup Q_{j_t-1})\cap D_{j_t}=\emptyset$. It
follows that in the iteration of the loop when $j=j_t$, $D_{j_t} \subseteq P$.
Thus, $P$ has value at least $c$ for $i$. Consequently, the body of the loop
would execute for $j=j_t$, a contradiction. Therefore, we have $(Q_1\cup
\ldots\cup Q_{j_t-1}) \cap D_{j_t}\neq \emptyset$.

Let $k$ be the smallest integer such that $3\leq k\leq j_t$ and
$(Q_1\cup\ldots\cup Q_{k-1})\cap D_{k}\neq \emptyset$. From our observation
above it follows that $k$ is well defined. Moreover, $(Q_1\cup\ldots \cup
Q_{k-2})\cap D_{k-1} =\emptyset$. Let $P$ be the path of unallocated goods
when $j=k-1$, that is, $P=P_0-(Q_1\cup\ldots\cup Q_{k-2})$. It follows that
$D_{k-1}$ is a subpath of $P$. Let $D$ be the shortest prefix of $P$ containing
$D_{k-1}$. It follows that $D$ is a strictly shorter prefix of $P$ than
$Q_{k-1}$ and $D$ has value at least $c$ to $i$. This is a contradiction with
the algorithm selecting $Q_{k-1}$ when $j=k-1$.
\end{proof}

%Extending the notation we introduced for splits, we call an allocation
%\emph{$q$-strong} if it assigns to each agent a bundle worth at least $q$
%to that agent. %We will prove the following theorem.
Extending the notation we used for splits, we call an allocation
\emph{$q$-strong} if it assigns to each agent a bundle worth at least
$q$ to that agent. We note that for regular sets of agents, $q$-strong and
$q$-sufficient allocations coincide.
%We will prove the following theorem.

\begin{theorem}\label{algorithm}
Let $P$ be a path of goods, $N=\{1,\ldots, n\}$, $n\geq 2$, a \emph{regular}
set of agents,
%each with a utility function on $P$ such that the total value of goods
%in $P$ equals $n$,
and $c$ a real such that $c\leq 1$. Further, let $Q$ be
a prefix of $P$ valued at least $c$ by at least one agent in $N$.
Let $R$ be the set computed by the algorithm $\alloc(N,P,Q,c)$ in line (2).
If no single good in $P$ has value at least $c$ for any agent in $N$ and
$|R|\leq \frac{1-c}{c}n +1$, then the algorithm $\alloc(N,P,Q,c)$ finds a
$c$-strong allocation of goods in $P$ to agents in $N$.
\end{theorem}
\begin{proof}
Let us assume that when the algorithm terminates, some agents are left without
a bundle.
Let us denote $s=|S|$, $r=|R|$, where $S$ and $R$ are computed in the lines (1)
and (2) of the algorithm $\alloc(N,P,Q,c)$ and let $k$ be the number of agents
that have no bundle when the algorithm terminates. By Proposition \ref{prop_S},
these agents are members of $R$. Moreover, by our assumption, $k\geq 1$.

Let $\ell$ be an unassigned agent. Clearly, in each iteration that assigns
a bundle to an agent from $S$, the value of that bundle is smaller than $c$ for
the agent $\ell$ (it is so because agents in $R$ have a preference over agents
in $S$ when bundles are assigned). Moreover, in each iteration $j$ when an
agent from $R$ is assigned a bundle, we recall this bundle is referred to as
$Q_j$, the value of $Q_j$ for the agent $\ell$ is smaller than $2c$.
Indeed, otherwise the prefix of $Q_j$ formed
by removing the last node of $Q_j$ would have value at least $c$ for
$\ell$ (since, by assumption, that node is worth less than $c$ to $\ell$).
This contradicts the property that $Q_j$ is a shortest prefix of the path
$P$ in the iteration $j$.

It follows that the total value for $\ell$ of all bundles constructed and
allocated by the algorithm is less than $cs + 2c(r-k)$. Consequently, the
value for $\ell$ of all unallocated goods when the algorithm terminates,
say $v_\ell$, satisfies $v_\ell> n-cs-2c(r-k)$ (since $N$ is a regular
collection of agents, for every agent in $N$, the total value of goods on
$P$ to that agent is $n$). On the other hand, by the stopping condition,
$v_\ell <c$. Thus, $c > n-cs-2c(r-k)$. Since $s=n-r$, it follows that
$r > \frac{1-c}{c}n+2k-1\geq\frac{1-c}{c}n+1$, a contradiction.

Thus, when the algorithm terminates, all agents are assigned bundles and
each bundle the algorithm allocates to an agent has value at least $c$
for that agent. In other words, the allocation defined by the algorithm is
$c$-strong.
\end{proof}

\begin{theorem}\label{any_number_or_agents}
Let $C$ be a cycle of goods and $N=\{1,\ldots, n\}$, where $n\geq 2$, a set
of agents, each with a utility function on $C$. Let
\[
c=\max_{d=n,n+1,\ldots}\min\left(\frac{n}{d},\frac{n}{\left\lceil\frac{n^2}{d}\right\rceil+n-2}\right)
\]
Then, there is a $c$-sufficient allocation of goods on $C$ to agents in $N$.
\end{theorem}
\begin{proof}
By Proposition \ref{prop2a}, it suffices to show the result under the
assumption that $N$ is a regular collection of agents.
%Under the assumptions we adopted, for every agent $i\in N$, $mms(i)=1$.
In particular, $c$-sufficient allocations and $c$-strong allocations coincide.

If some node $x$ of $C$ has value at least $c$ to some agent $j$ then
a $c$-sufficient allocation exists. Indeed, we assign $x$ to the agent $j$.
By Proposition \ref{prop_single}, the values of the maximin shares for the
remaining $n-1$ agents and the path $C-x$ do not drop. Applying
the algorithm of Bouveret {\it at al.} \cite{BouveretCEIP17}, we construct
an mms-allocation of goods on the path $C-x$ to those $n-1$ agents. This
allocation together with an assignment of $x$ to the agent $j$ is a
$c$-sufficient allocation of
goods on $C$ to agents in $N$ ($j$ receives a bundle worth at least at $c$
and all other agents receive bundles worth at least their maximin share).
So, we assume that no single-element bundle is $c$-sufficient for any agent
in $N$.

For every agent $i\in N$ we select any of her mms-splits of $C$. By regularity,
all bundles in these splits are non-empty. Thus, each selected split can be
obtained by removing $n$ different edges, say $e_1^i,\ldots,e_n^i$, from $C$.
We arrange all edges $e_j^i$, $i,j=1,\ldots, n$, into a sequence $E=e_0,e_1,
\ldots, e_{n^2-1}$. To this end, we start in any place in $C$ and inspect the
edges of $C$ moving clockwise. Each time we find an edge $e$ used to obtain
mms-splits for, say, $k$ agents, we place $k$ occurrences of $e$ in the
sequence.

For an integer $d=n,n+1,\ldots$, we define
\[
f(d)=\min\left(\frac{n}{d},\frac{n}{\left\lceil\frac{n^2}{d}\right\rceil+n-2}
\right).
\]
We then define $p$ to be that integer $d\geq n$, for which $f(d)$ achieves
its maximum (in case of ties we pick for $p$ the smallest of those values $d$).
Since $f(n)> 1/2$ and, for $d\geq n^2$, $f(d) \leq 1/n $, it
follows that $n\leq p< n^2$. It is also clear that $f(p)=c$ (as defined in the
statement of the theorem). Moreover, $c\leq\frac{n}{\left\lceil n^2/p\right\rceil+n-2}\leq 1$.
%and define $c=f(d)$.

Let us define $\overline{h}=\left\lceil\frac{n^2}{p}\right\rceil$,
$\underline{h}=\left\lfloor\frac{n^2}{p}\right\rfloor$ and $r=n^2-p\cdot
\underline{h}$. Clearly, we have $\underline{h} \geq 1$ and $0\leq r < p$.
We define a split of the cycle $C$ into $p$ parts (some of them possibly
empty) by removing $p$ edges $e_i$, where $i=0,\overline{h},2\overline{h},
\ldots, r\overline{h},r\overline{h}+\underline{h},r\overline{h}+2\underline{h},
\ldots, r\overline{h}+(p-r-1)\underline{h}$. It is easy to check that between
two consecutive removed edges there are $\overline{h}-1$ or $\underline{h}-1$
edges of the sequence $E$.

Let $Q$ be a part of this split with the largest value to agent $n$
(any other agent could be chosen for $n$, too). This value is at least
$\frac{n}{p} \geq f(p)=c$. Let $P$ be the path obtained from $C$ by removing
an edge so that $Q$ is a prefix of $P$. We will show that the call
$\alloc(N,P,Q,c)$ produces a $c$-strong allocation for $N$.

First, we note that $Q$ contains at most $\overline{h}-1$ edges of the
sequence $E$.
This means that there are at most $\overline{h}-1$ agents such that $Q$
intersects more than one part of their mms-split. For each of the remaining
agents their mms-split gives rise to an $(n-1)$-split of the path $P-Q$
that is $1$-strong for them. Since $c\leq 1$, these splits are $c$-strong
for the
corresponding agents and, consequently, all these agents are in $S$ --- the
set defined in line (1) of the algorithm $\alloc(N,P,Q,c)$.
It follows that $|S| \geq n - (\overline{h}-1)$. Let
$R$ be the set defined in line (2) of $\alloc(N,P,Q,c)$. Clearly,
$|R|=n-|S|\leq \overline{h}-1$. Hence, by
the definition of $p$,
\[
|R|\leq\overline{h}-1 = \left\lceil\frac{n^2}{p}\right\rceil-1 \leq \frac{n}{f(p)}-n+1 = \frac{n}{c}-n+1 = \frac{1-c}{c}n+1.
\]
By Theorem \ref{algorithm}, there is a $c$-strong allocation for $P$. This
allocation is also a $c$-strong allocation for $C$.
\end{proof}

\medskip
This result yields a corollary that displays a specific value $c$ for which
the existence of $c$-sufficient allocations of goods on cycles is guaranteed.
Let
\[
\varphi=\frac{\sqrt{5}+1}{2}\quad\mbox{and}\quad \psi=\frac{1}{\varphi}=\frac{\sqrt{5}-1}{2}\approx 0.61803..\ .
\]

\begin{corollary}\label{cor:approx}
For any number $n$ of agents there is a $\psi$-sufficient allocation for a
cycle.
\end{corollary}
\begin{proof} Clearly, the corollary holds for $n=1$, so we assume that
$n\geq 2$. Since $\left\lfloor\varphi n\right\rfloor \geq n$, by Theorem
\ref{any_number_or_agents}, it suffices to show that
\[
\min\left(\frac{n}{d},\frac{n}{\left\lceil\frac{n^2}{d}\right\rceil+n-2}\right) \geq \psi,
\]
where $d=\left\lfloor\varphi n\right\rfloor$.

Clearly,
$
\frac{n}{\left\lfloor\varphi n\right\rfloor}\geq\frac{n}{\varphi n}=\psi
$. Thus, it remains to prove that
\[
\psi \leq \frac{n}{\left\lceil\frac{n^2}{\lfloor n\varphi\rfloor}\right\rceil
+n-2}.
\]
To this end, we observe that $n^2\leq n^2+n-1=(\varphi n-1)(\psi n+1)$, so
\[
\psi\cdot\left(\left\lceil\frac{n^2}{\left\lfloor\varphi n\right\rfloor}\right\rceil+n-2\right) \leq \psi\cdot\left(\left\lceil\frac{(\varphi n-1)(\psi n+1)}{\varphi n-1}\right\rceil+n-2\right)
\]
\[
\leq \psi\cdot(\psi n+2+n-2)=\psi(\psi+1)n=n.
\]
\end{proof}

We now turn our attention to the problem of $c$-sufficient allocations when
some agents have the same utility function, that is, are of the same type.
Specifically, we will consider allocations of goods on cycles to $n$ agents
of at most $t$ types.

\begin{lemma}\label{lem_t_types}
Let $C$ be a cycle of goods, $N$ a regular set of agents, $t\geq 2$ an
integer, and $N_i$, $1\leq i\leq t$, pairwise disjoint subsets of $N$ such
that $\bigcup_{i=1}^t N_i=N$, $|N_1|\geq |N_2|\geq \ldots\geq |N_t|$ and,
for every $i=1,\ldots, t$, all agents in $N_i$ are of the same type. If
$N_1,N_2\neq\emptyset$ and there are splits $\Pi_1$ and $\Pi_2$ such that
$\Pi_i$ is $\frac{t}{2t-2}$-strong for agents in $N_i$, for $i=1,2$, and
there is a bundle in $\Pi_1$ and a bundle in $\Pi_2$ whose intersection has
value at
least $\frac{t}{2t-2}$ to some agent in $N$, then there is a
$\frac{t}{2t-2}$-strong allocation of goods on $C$ to agents in $N$.
\end{lemma}
\begin{proof} Let us define $n=|N|$, $c_t=\frac{t}{2t-2}$, and $n_i=|N_i|$,
$1\leq i\leq t$. We will call agents in a set $N_i$, $1\leq i\leq t$, to be
of type $i$.
%\red{If $N_2=\emptyset$, then all agents are of the same type
%and an mms-allocation exists (cf. Proposition \ref{prop3}). Since the agents
%are regular, for every agent $i\in N$, $\mms{(n)}(i)=1\geq c_t$. Thus, that
%mms-allocation is $c_t$-strong. From now on we will assume that $N_2\neq
%\emptyset$. Consequently, we also have that $N_1\neq emptyset$.}

%Further, we note that by Proposition \ref{prop2a}, we may assume
%that the agents are proportional and that for every agent the total value of
%all nodes of $C$ is $n$. In such case, the maximin share for each agent
%is 1, and $c$-strong and $c$-sufficient allocations coincide. Thus, to prove
%the assertion it suffices to show the existence of a $c_t$-strong allocation.

If there is a single good, say $x$, in $C$ of value at least $c_t$ to
some agent $k\in N$, then agent $k$ receives $x$. It follows from Theorem
\ref{bouv-alloc} that there is an mms-allocation for the path $C-x$ and the
remaining $n-1$ agents. We distribute the goods of $C-x$ to these $n-1$
agents according to this allocation. By Proposition \ref{prop_single},
each agent $i\in N\setminus\{k\}$ receives a bundle worth to her at least
$\mms^{(n)}(i)$. Since each agent is regular, for every $i\in N$ we have
$\mms^{(n)}(i)=1\geq c_t$. Thus, each agent in $N$ is allocated a bundle
worth at least $c_t$ to her. In other words, the allocation we constructed is
$c_t$-strong.

Thus, let us assume that no single good in $C$ has value at least $c_t$
for any agent in $N$. Let $Q$ be the intersection of two bundles $A$ and $B$,
where $A$ is a bundle from $\Pi_1$ and $B$ is a bundle from $\Pi_2$, such
that $Q$ has value at least $c_t$ for some agent in $N$. It is clear that
$Q$ is a proper subpath of $C$. Let $P$ be the path obtained from $C$ by
removing an edge so that $Q$ is a
prefix of $P$. We will consider the call $\alloc(N,P,Q,c_t)$ and follow the
notation introduced in the description of the algorithm.

Since $Q\subseteq A,B$, all agents of types 1 and 2 are in $S$. It follows
that $|R|\leq n-(n_1+n_2)$. Since $n_1\geq n_2\geq \ldots \geq n_t$,
$n_1+n_2\geq 2n_j$, for $j=3,\ldots,t$. Thus,
\[
(t-2)(n_1+n_2)\geq 2(n_3+\ldots +n_t)=2(n-(n_1+n_2)).
\]
Consequently, we have $n_1+n_2\geq\frac{2}{t}n$. We use this inequality
to estimate $|R|$ getting
\[
|R|\leq n-\frac{2}{t}n=\frac{t-2}{t}n=\frac{1-c_t}{c_t}n.
\]
By Theorem \ref{algorithm}, there is a $c_t$-strong allocation of goods on $C$
to agents in $N$.
\end{proof}

We will use this lemma to obtain a general result about the existence of
$c$-sufficient allocations for any number of agents of $t\geq 4$ types.
Afterwards, we will obtain results for the two specific cases of $t=2$ and
$t=3$.

\begin{theorem}\label{any_number_types_theorem}
Let $C$ be a cycle of goods, $N$ a set of agents of at most $t$ types, where
$t\geq 4$. Then, a $\frac{t}{2t-2}$-sufficient allocation of goods on $C$ to
agents in $N$ exists.
\end{theorem}
\begin{proof} By Proposition \ref{prop2a} we may assume that all agents
in $N$ are regular. Thus, we have that the maximin share for all agents is
$1$. %In such case, $c$-strong and $c$-sufficient allocations coincide.

Let $n_1,\ldots,n_t$ be the numbers of agents of types $1,\ldots,t$,
respectively, and let $n=n_1+\ldots+n_t$ be the number of all agents in
$N$. As before, we define $c_t=\frac{t}{2t-2}$ and assume that $n_1\geq n_2
\geq\ldots\geq n_t$. If all agents are of the same type, Proposition
\ref{prop3} ensures the existence of an mms-allocation. Since $c_t < 1$,
this mms-allocation is also a $\frac{t}{2t-2}$-sufficient allocation. Thus,
we assume that $n_2 >0$ (and so, obviously, $n_1>0$,
as well).

Let $A_1,\ldots,A_n$ and $B_1,\ldots,B_n$ be mms-splits of the cycle
for agents of type $1$ and $2$, respectively. In particular, since $c_t<1$,
the two splits are $c_t$-strong for agents of type 1 and 2, respectively.

If for some $i$ and $j$, where $1\leq i,j\leq n$, the path $A_i\cap B_j$
has value at least $c_t$ to some agent in $N$, then we are done by
Lemma \ref{lem_t_types}. Indeed, by the regularity of the set of agents,
$c_t$-strong and $c_t$-sufficient allocations coincide.

So, let us assume that no path $A_i\cap B_j$ has value $c_t$ or more to any
agent. Then, in particular, $A_i\not\subseteq B_j$ and $B_j\not\subseteq A_i$
for all $i$ and $j$, $1 \leq i,j\leq n$. Thus, each of the sets $A_1,\ldots,
A_n$ intersects exactly two consecutive sets of the mms-split $B_1,\ldots,
B_n$. We can assume without loss of generality that for every $i$, $1\leq
i\leq n$, the set $A_i$ intersects the sets $B_i$ and $B_{i+1}$ (the
arithmetic on indices is modulo $n$, adjusted to the range $[1..n]$).

We claim that the mms-split $A_1,\ldots,A_n$ for agents of type $1$ is
a $c_t$-sufficient split for agents of type $2$. To prove it denote by
$u_2$ the utility function for agents of type $2$. Since for every $i$,
$1\leq i\leq t$, $A_i\subseteq B_i \cup B_{i+1}$, it follows that
\[
2=u_2(B_i\cup B_{i+1})=u_2(A_{i-1}\cap B_i)+u_2(A_i)+u_2(A_{i+1}\cap B_{i+1}).
\]
By our assumption the sets $A_{i-1}\cap B_i$ and $A_{i+1}\cap B_{i+1}$ have
value less than $c_t$ to any agent. Since for $t\geq 4$ we have $c_t\leq
\frac{2}{3}$, it follows that $u_2(A_i)>2-2c_t\geq c_t$. We proved that
$A_1,\ldots,A_n$ is a $c_t$-strong split for any agent of type $2$. As we
noted, it is also a $c_t$-strong split for agents of type 1. Therefore,
we can take $A_1,\ldots, A_n$ for $\Pi_1$ and $\Pi_2$ in Lemma
\ref{lem_t_types}. Moreover, $A_1$ is clearly the intersection of a bundle
in $\Pi_1$ with a bundle in $\Pi_2$ simply because $A_1$ is a bundle in each
of these splits. Moreover, as $A_1,\ldots,A_n$ is an mms-split for agents
of type 1, $A_1$ has value at least 1 and, consequently, at least $c_t$ for
agents of type 1. Thus,
the assumptions of Lemma \ref{lem_t_types} are satisfied and a $c_t$-strong
allocation exists. This completes the proof as $c_t$-strong and
$c_t$-sufficient allocations coincide.
\end{proof}

Next, we consider the case of agents of two types and show that
$\frac{3}{4}$-sufficient allocations of goods on a cycle to $n$ agents of
two types are always possible. Later on we will present a result showing
that $\frac{3}{4}$-sufficient allocations of goods on a cycle also exist when
agents are of three types. The proof of this more general result is highly
technical, complex and long (enough so that we provide it in an appendix).
In contrast, the proof in the case of agents of two types is simple yet
still of substantial interest in its own right to be presented.

\begin{theorem}\label{two_types_theorem}
Let $C$ be a cycle of goods and $N$ a set of agents of no more than two types.
Then, a $\frac{3}{4}$-sufficient allocation of goods on $C$ to agents in $N$
exists.
\end{theorem}
\begin{proof} By Proposition \ref{prop2a} we may assume that the set $N$
of agents is regular. In particular, $\frac{3}{4}$-strong and
$\frac{3}{4}$-sufficient allocations coincide. Let $n$ be the number of
agents and let $n_1$ and $n_2$ be the numbers of agents of type 1 and type
2, respectively. Clearly, $n=n_1+n_2$. We will write $u_1$ and $u_2$ for
the utility functions for the agents of type 1 and type 2, respectively.
Without loss of generality, we will assume that $n_1\geq n_2$. If $n_2=0$,
then all agents are of the same type and an mms-allocation (or, equivalently,
a 1-sufficient allocation) exists by Proposition \ref{prop3}. This implies
the assertion. Therefore, we will assume that $n_2>0$. Consequently, $n_1>0$,
too.

Let $A_1,\ldots,A_n$ and $B_1,\ldots,B_n$ be mms-splits for agents of
type $1$ and type $2$, respectively. If some set $A_i\cap B_j$, where $1\leq
i,j\leq n$, is of value at least $\frac{3}{4}$ to some agent, then we allocate
this set $A_i\cap B_j$ to this agent. Let $P=C-(A_i\cap B_j)$. Clearly, $P$
is a path. Moreover, since every bundle in the split $A_1,\ldots, A_n$ other than $A_i$
is included in $P$, $\mms^{(n-1)}(P,u_1) \geq \mms^{(n)}(C,u_1)=1$. Similarly,
we have $\mms^{(n-1)}(P,u_2)\geq 1$. Thus, we can assign the goods in $P$ to
the remaining $n-1$ agents so that each agent receives a bundle that she values
at least at $1$. To this end, we may use the algorithm by Bouveret
\emph{et al.}\cite{BouveretCEIP17} (cf. Theorem \ref{bouv-alloc}). The
resulting allocation of goods on $C$ is $\frac{3}{4}$-strong and so,
$\frac{3}{4}$-sufficient.

Therefore, we assume from now on that no set $A_i\cap B_j$, $1\leq i,j\leq n$,
is of value $\frac{3}{4}$ or more for any of the agents. In particular,
$A_i\not\subseteq B_j$ and $B_j\not\subseteq A_i$ for all $i$ and $j$, $1\leq i,j\leq n$.
Thus, without loss of generality, we can assume that each
set $A_i$ has a nonempty intersection with the sets $B_i$ and $B_{i+1}$
(addition modulo $n$ adjusted for $[1..n]$) and with no other set $B_j$.

We will show that the sets  $A_1,\ldots,A_n$ can be allocated so that each
agent receives a set of value at least $\frac{3}{4}$ to this agent. Suppose
that fewer than $\frac{n}{2}$ of the sets $A_1,\ldots,
A_n$ have value at least $\frac{3}{4}$ to agents of type 2. Then, there are
two sets $A_i,A_{i+1}$ such that each of them is valued at less than
$\frac{3}{4}$ by agents of type 2. Thus, we have
\begin{align*}
3&=u_2(B_i\cup B_{i+1}\cup B_{i+2})\\
 &=u_2(A_{i-1}\cap B_i)+u_2(A_{i}\cup A_{i+1})+u_2(A_{i+2}\cap B_{i+2})\\
 &<\frac{3}{2}+u_2(A_{i-1}\cap B_i)+u_2(A_{i+2}\cap B_{i+2}).
\end{align*}
This inequality implies that at least one of the sets $A_{i-1}\cap B_i$ and
$A_{i+2}\cap B_{i+2}$ has value larger that $\frac{3}{4}$ for agents of type
2, a contradiction.

It follows that at least $\frac{n}{2}$ of the sets $A_1,\ldots,A_n$ are of
value at least $\frac{3}{4}$ to agents of type 2. Since there are at most
$\frac{n}{2}$ agents of this type (we recall that $n_1\geq n_2$), we have
sufficiently many such sets for them. The remaining sets are allocated to
agents of type 1. The resulting allocation is $\frac{3}{4}$-strong and so,
also $\frac{3}{4}$-sufficient.
\end{proof}

Theorem \ref{two_types_theorem} is tight. Let us consider a cycle $v_1,v_2,\ldots, v_{12}$, and six agents of two types with the utility functions shown in
Figure \ref{fig_2_types}. It is easy to check that in this case
$\frac{3}{4}$-sufficient allocation exists but $c$-sufficient allocation
for $c>\frac{3}{4}$ does not.

\begin{figure*}[ht!]
{\small
\centerline{
\begin{tabular}{|l|c|c|c|c|c|c|c|c|c|c|c|c|}
  \hline
 & $v_1$ & $v_2$ & $v_3$ & $v_4$ & $v_5$ & $v_6$ & $v_7$ & $v_8$ & $v_9$ & $v_{10}$ & $v_{11}$ & $v_{12}$\\
  \hline
agents 1, 2, 3 & 3 & 3 & 1 & 2 & 2 & 1 & 3 & 3 & 1 & 2 & 2 & 1 \\
  \hline
agents 4, 5, 6 & 3 & 1 & 2 & 2 & 1 & 3 & 3 & 1 & 2 & 2 & 1 & 3 \\
  \hline
\end{tabular}}
\caption{An example showing that the result in Theorem \ref{two_types_theorem}
is sharp.}\label{fig_2_types}
}
\end{figure*}

As we pointed out above, the $\frac{3}{4}$ fraction of maximin shares can
be guaranteed in a more general setting when agents are of three types.
Specifically, the following theorem holds. We provide its proof in the
appendix.

\begin{theorem}\label{three_types_theorem}
Let $C$ be a cycle of goods and $N$ a set of agents of at most three types.
Then, a $\frac{3}{4}$-sufficient allocation of goods on $C$ to agents in $N$
exists.
\end{theorem}

Theorem \ref{three_types_theorem} is tight too. In Figure \ref{fig_3/4_for_3_types} we present an example where a $\frac{3}{4}$-sufficient allocation of goods on a cycle exists but there is no $c$-sufficient allocation for any $c>\frac{3}{4}$.

\begin{figure*}[ht!]
{\small
\centerline{
\begin{tabular}{|l|c|c|c|c|c|c|c|c|c|c|c|c|c|c|c|c|c|c|}
  \hline
 & $v_1$ & $v_2$ & $v_3$ & $v_4$ & $v_5$ & $v_6$ & $v_7$ & $v_8$ & $v_9$ & $v_{10}$ & $v_{11}$ & $v_{12}$ & $v_{13}$ & $v_{14}$ & $v_{15}$ & $v_{16}$ & $v_{17}$ & $v_{18}$\\
  \hline
agents 1, 2 & 2 & 0 & 2 & 1 & 2 & 1 & 2 & 0 & 2 & 1 & 2 & 1 & 2 & 0 & 2 & 1 & 2 & 1  \\
  \hline
agents 3, 4 & 2 & 1 & 2 & 1 & 2 & 0 & 2 & 1 & 2 & 1 & 2 & 0 & 2 & 1 & 2 & 1 & 2 & 0 \\
  \hline
agents 5, 6 & 2 & 1 & 2 & 0 & 2 & 1 & 2 & 1 & 2 & 0 & 2 & 1 & 2 & 1 & 2 & 0 & 2 & 1 \\
  \hline
\end{tabular}}
\caption{An example showing that the result in Theorem \ref{three_types_theorem}
is sharp.}\label{fig_3/4_for_3_types}
}
\end{figure*}

An interesting property of the example in Figure \ref{fig_3/4_for_3_types} is that the set $\{0,1,2\}$ of values of the utility functions is very small. The problem of existence and construction of an mms-allocation with this set of values was studied in the original version of the mms-allocation problem (i.e. for complete graphs in our terminology). Amanatidis {\it et al.}  \cite{AmanatidisMNS15} proved that unlike in the case of a cycle (see Figure \ref{fig_3/4_for_3_types}) for a complete graph an mms-allocation with $\{0,1,2\}$ as the set of values of the utility functions always exists.

It is also easy to show that if $\{0,1\}$ is the set of values of the utility functions, then an mms-allocation for a cycle always exists and can be constructed in polynomial time. An analogous statement for complete graphs was observed earlier by Bouveret and Lema{\^{\i}}tre \cite{ BouveretL16}.

As we have seen earlier, even for three agents mms-allocations of goods on
a cycle may not exist. Theorem \ref{three_types_theorem} implies that when allocating
goods on a cycle to three agents, we can guarantee that each agent receives
a bundle worth at least $\frac{3}{4}$ of her maximin share. We will now
show that this guarantee can be strengthened to $\frac{5}{6}$. Re-examining
the example in Figure \ref{fig_no_mms_for_3} shows that this is the best
guarantee we can get. To see this, we recall that in this case the maximin
share value for agents 1 and 2 is 5 and for agent 3 is 6. In particular,
any split in which an agent 1 or 2 obtains a bundle consisting of two or
fewer items is at best a $\frac{4}{5}$-sufficient allocation (indeed, any
two consecutive items have the total value of no more than 4). Thus, in
any $c$-sufficient allocation with $c \geq \frac{5}{6}$, agents 1 and 2
receive bundles of at least three items. If agent 3 receives at least three
items, then all agents obtain bundles of exactly three items. There are only
three such allocations. It is easy to see that one of them is
$\frac{4}{5}$-sufficient and the other two are $\frac{5}{6}$-sufficient.
Since any two consecutive items have the total value at most 5 for agent
3, any allocation in which agent 3 receives no more than 2 items is at best
$\frac{5}{6}$-sufficient (some of them are actually $\frac{5}{6}$-sufficient;
for instance, the allocation $\{v_4,v_5, v_6\},$ $\{v_7,v_8,v_9,v_1\},
\{v_2,v_3\}$.

To show that a $\frac{5}{6}$-sufficient allocation of goods on a cycle always
exists for $3$ agents we start by introducing some additional terminology and
two lemmas.
Let $C$ be a cycle, $N=\{1,2,3\}$ a set of agents and $u_i$ the utility
function of an agent $i=1,2,3$. To simplify notation, we call an mms-split
for $(C,3,u_i)$ an \emph{$\mms(i)$-split} and we recall that we write
$\mms(i)$ as a shorthand for $\mms^{(3)}(C,u_i)$.

\begin{lemma}\label{prop8b}
Let $C$ be a cycle and $N=\{1,2,3\}$ a set of agents.
%with the utility functions $u_1,u_2$ and $u_3$ on $C$, respectively.
Let $c$, $0<c\leq 1$, be a real number. If for some agent $i\in N$
two different bundles
of an $\mms(i)$-split of $C$ are of value at least $c\cdot mms(j)$ to an agent
$j$, where $j\neq i$, then there is a $c$-sufficient allocation of goods on
$C$ to agents in $N$.
\end{lemma}

\begin{proof}
We assume without loss of generality that $i=1$ and $j=2$. Let us observe
that for any agent and any $3$-split of the cycle there is a bundle of value
at least one third of the total value of all goods for this agent. The value
of this bundle is larger than or equal to the maximin share for this agent.
In particular, the value of one of the bundles of any $mms(1)$-split is greater
than or equal to $\mms(3)$. The following protocol finds a $c$-sufficient
allocation. Agent 3 picks a bundle of the $mms(1)$-split that has value
at least $\mms(3)$ to her. At least one of the remaining two bundles has value
at least $c\cdot mms(2)$ to agent $2$. That bundle is allocated to agent 2.
The remaining bundle of the split is assigned to agent $1$.
\end{proof}

\begin{lemma}\label{prop8a}
Let $C$ be a cycle and $N=\{1,2,3\}$ a set of agents.
%with the utility functions $u_1,u_2$ and $u_3$ on $C$, respectively.
Let $c$, $0<c\leq 1$, be a real number. If for some two different agents
$i,j\in N$, the intersection of a bundle of an $mms(i)$-split and a bundle
of an $mms(j)$-split of $C$ has value at least $c\cdot mms(i)$ to the agent
$i$, then there is a $c$-sufficient allocation of goods in $C$ to agents in
$N$.
\end{lemma}

\begin{proof}
Without loss of generality, we assume that $i=1$ and $j=2$.
Let $A$ and $B$ be bundles of mms-splits for agents $1$ and $2$,
respectively, such that $Q=A\cap B$ has value at least $c\cdot mms(1)$
to the agent $1$. The set $B-Q$ consists of one or two intervals, the latter
in the case when $A\subset B$. Clearly, the elements in $B-Q$ can be
added to the remaining two bundles (different from $B$) of the
$\mms(2)$-split so that to form a $2$-split of $C-Q$, say $B',B''$, with
both parts of value at least $\mms(2)$ for the agent $2$.

We construct a $c$-sufficient allocation as follows. If
$u_3(Q)\geq c\cdot\mms(3)$, then we allocate $Q$ to agent 3. Since
$Q\subseteq A$, $C-Q$ contains the remaining two bundles (different from
$A$) of the $\mms(1)$-split. Hence, $u_1(C-Q)\geq 2\cdot\mms(1)$. Since
$C-Q=B'\cup B''$, at least one of the bundles $B'$ and $B''$ is worth
$\mms(1)$ or more to the agent 1. We allocate that bundle to that agent
and we allocate the remaining one to the agent 2.

If $u_3(Q) < c\cdot\mms(3)$, then the agent $1$ gets $Q$. Since
$u_3(C)\geq 3\cdot\mms(3)$ and $c\leq 1$, $u_3(C-Q)> 2\cdot \mms(3)$.
Thus, at least one bundle of $B'$ and $B''$ is worth $\mms(3)$ or more to agent $3$. That bundle goes to
the agent 3 and the remaining one to agent 2.
\end{proof}

%\medskip
%The following statement is a special case of Lemma \ref{prop8a} when $c=1$ and the two intersecting parts of mms-splits are contained one in the other.
%\begin{corollary}\label{cor2}
%If there are two different agents $i,j\in\{1,2,3\}$ such that a part of an mms$_i$-split of a cycle is contained in a part of an mms$_j$-split of this cycle, then there is an mms-allocation for the cycle and these 3 agents. $\hfill\Box$
%\end{corollary}

We are now ready to prove that if goods on a cycle are allocated to three
agents, then each agent can be guaranteed a bundle worth at least $\frac{5}{6}$
of her maximin share.

\begin{theorem}\label{th2}
Let $C$ be a cycle and $N=\{1,2,3\}$ a set of agents. Then, a
$\frac{5}{6}$-sufficient allocation of goods in $C$ to agents in $N$ exists.
\end{theorem}
\begin{proof} Because of Proposition \ref{prop2a}, we restrict attention to the
case when the set $N$ of agents is regular. It follows that each agent values
all goods on $C$ at 3 and that all agents are proportional. In particular,
$\mms(i)=1$, for every agent $i\in N$.

Let $\Pi_1=A_1,A_2,A_3$, $\Pi_2=B_1,B_2,B_3$ and $\Pi_3=C_1,C_2,C_3$
be mms-splits for agents $1$, $2$ and $3$, respectively. If a bundle of one
of these splits is contained in a bundle of another split, then an
mms-allocation (and consequently a $\frac{5}{6}$-sufficient allocation)
exists, by Lemma \ref{prop8a} (applied to the two bundles and $c=1$).

From now on we assume that there are no such containments. Then, we
can relabel the parts of the splits $\Pi_1$, $\Pi_2$ and $\Pi_3$ so that
the sets
$A_{\lceil (k+2)/3\rceil}\cap B_{\lceil(k+1)/3\rceil}\cap
C_{\lceil k/3\rceil}$, $k=1,2,\ldots,9$ (indices are computed modulo
$3$ with respect to the set $\{1,2,3\}$) form a split of the cycle. Let us call these 9 sets {\it chunks}.

Since each agent $i\in N$ values $C$ at 3, any split of $C$ into three bundles
contains at least one bundle of value at least $1$ to $i$, that is, of value
at least $\mms(i)$ to $i$. In particular, such a bundle can be found in the
mms-splits of the two agents other than $i$. Moreover, if for some two agents
$i,j\in N$, $i\not=j$, more than one part in $\Pi_i$ has value at least $1$
($=\mms(j)$) for $j$, then we are done by Lemma \ref{prop8b}.

Let $i,j,k\in N$ be three different agents. We observe that if a bundle of
$\Pi_k$ of value at least $1$ ($=\mms(i)$) for the agent $i$ is different
from a part of $\Pi_k$ of value at least $1$ ($=\mms(j)$) for the agent $j$,
then we can easily distribute the parts of $\Pi_k$ among the agents to get
an mms-allocation.

So, we assume that for every agent $i$ there is a unique bundle in $\Pi_i$
whose value for each of the remaining two agents is at least $1$. Let us call
this bundle {\it $i$-significant}.

One can readily observe that, for $i,j\in N$, $i\not= j$, the intersection of
the $i$-significant bundle and the $j$-significant bundle consists of either
two chunks, one chunk or is empty. Moreover, it is not possible for each of
the three pairs of different significant bundles to intersect on one chunk.
Thus, there is a pair $i,j\in N$ of agents such that the $i$-significant bundle
and the $j$-significant bundle are disjoint or intersect on two chunks. We will
consider these two cases separately.  Clearly, we can assume without loss of
generality that $i=1$, $j=2$ and that $A_s$ and $B_t$, $s,t\in\{1,2,3\}$, are
the $1$-significant and $2$-significant bundles, respectively,

\smallskip
\noindent
{\bf Case 1} $A_s$ and $B_t$ are disjoint.

\smallskip
We will not lose generality, either, if we assume that $A_s=A_1$ and
$B_t=B_2$. Let us consider the mms-split $\Pi_3$. By the definition of
significant bundles, the value for the agent $3$ of each of the sets $A_1$
and $B_2$ is at least $1$.

If the value of $A_2$ or $A_3$ for the agent $3$ is at least $\frac{5}{6}$,
then a $\frac{5}{6}$-sufficient allocation exists by Lemma \ref{prop8b}.
Moreover, if the value for the agent $3$ of the set $A_1\cap C_3$ is at
least $\frac{5}{6}$, then a $\frac{5}{6}$-sufficient allocation exists by
Lemma \ref{prop8a}.

So, let us assume that the values for the agent $3$ of the sets $A_2$, $A_3$
and $A_1\cap C_3$ are all smaller than $\frac{5}{6}$. Then the value of the
chunk $A_1\cap B_1\cap C_1$ is larger than $3-3\cdot\frac{5}{6}=\frac{1}{2}$.

In a very similar way we show that if any of the sets $B_1,B_3,B_2\cap C_2$
has value at least $\frac{5}{6}$ for the agent $3$, then a
$\frac{5}{6}$-sufficient  allocation exists. Otherwise, the value for the
agent $3$ of the chunk $A_2\cap B_2\cap C_1$ is larger than $\frac{1}{2}$.

Then, however, the value for the agent $3$ of the part $C_1$, which contains
the chunks $A_1\cap B_1\cap C_1$ and $A_2\cap B_2\cap C_1$ is larger than
$1$, a contradiction because the agents are proportional.

\smallskip
\noindent
{\bf Case 2} $A_s$ and $B_t$ intersect on two chunks.

\smallskip
We assume without loss of generality that $A_s=A_1$ and $B_t=B_1$. By the
definition of significant bundles, the value for the agent $3$ of each of
the bundles $A_1$ and $B_1$ is at least $1$.

Reasoning in exactly the same way as in the first case, we conclude that
either a $\frac{5}{6}$-sufficient allocation exists or the value for the agent $3$ of the
chunk $A_1\cap B_1\cap C_1$ is larger than $\frac{1}{2}$.

Let us then assume the latter. We observe that as in Case 1, if any of the
bundles $B_2,B_3,B_1\cap C_1$ has value at least $\frac{5}{6}$ for the agent
$3$, then a $\frac{5}{6}$-sufficient allocation exists (by Lemma \ref{prop8b}
or \ref{prop8a}). Otherwise, the value for the agent $3$ of the chunk $A_1\cap
B_1\cap C_3$ is larger than $\frac{1}{2}$.

Then, however, the value for the agent $3$ of the set $A_1\cap B_1=(A_1\cap B_1
\cap C_1)\cup (A_1\cap B_1\cap C_3)$ is larger than $1$. In this case an
mms-allocation exists because we can assign $A_2$ to the agent $1$, $B_3\cup
(A_3\cap B_2\cap C_2)$ to the agent $2$ and $A_1\cap B_1$ to the agent $3$.
\end{proof}

For more than three agents a $\frac{5}{6}$-sufficient allocation of goods on a cycle may not exist. The example in Figure \ref{fig_2_types} shows that this may
be the case even for agents of two types only. For this example, we have
already observed
that there is a $\frac{3}{4}$-sufficient allocation but no $c$-sufficient allocation for $c>\frac{3}{4}$.

We close this section by commenting on the problem of computing
$c$-sufficient allocations. First, we note that all proofs we presented in
this section, as well as the proof of Theorem \ref{three_types_theorem}, which
we
give in the appendix, yield algorithms for constructing $c$-sufficient
allocations.
These constructions
apply to the case of regular sets of agents and require that both maximin
shares and mms-splits be computed for every agent. Given those, the actual
construction of a $c$-sufficient allocation provided in each proof can
clearly be implemented to run in polynomial time. We recall that in the
case of goods on a cycle, maximin shares and mms-splits can be computed in
polynomial time. It follows then by Corollary \ref{regular} that in all
cases we discussed in this section when $c$-sufficient allocations exist,
they can be computed by algorithms running in polynomial times.

\section{Conclusions}

%In this paper
We investigated maximin share allocations in the
graph setting for the fair division problem of indivisible goods proposed by
Bouveret \emph{et al.} \cite{BouveretCEIP17}.
That paper settled the case of trees by showing
that maximin share allocations of goods on trees always exist and can be
computed in polynomial time. It also gave an example of goods on a cycle to
be distributed to four agents where no maximin share allocation exists.

Our work focused on cycles. We found several cases when maximin share
allocations on cycles exist and can be found in polynomial time. For some
other cases, when maximin share allocations are not guaranteed to exist,
we found polynomial-time algorithms deciding the existence of maximin share
allocations and, if they do exist, computing them, too. Interestingly, we do
not know the complexity of deciding the existence of maximin share allocations
on cycles. We proved that the problem is in NP but whether it is NP-hard is
open.

In general, understanding the complexity of deciding the existence of maximin
share allocations of goods on graphs is a major challenge. We improved an
earlier upper bound from $\Sigma_2^P$ down to $\Delta_2^P$ but, as in the case
of cycles, we do not have any hardness results. Establishing such results and
characterizing classes of graphs for which deciding the existence of maximin
share allocations is in P, is NP-complete or goes beyond NP (under
the assumption, of course, that the polynomial hierarchy does not collapse) are important open problems.

Perhaps our most interesting results concern the existence of allocations
guaranteeing each agent a given fraction of her maximin share. For instance,
we show that for three agents one can always find an allocation giving each
agent at least $5/6$ of her maximin share. For an arbitrary number of agents
of two or three types, we show allocations giving each agent $3/4$ of the maximin
share, and we also obtain some results for the case when the number of
types is larger than three. Moreover, in each case, these allocations can
be found by polynomial-time algorithms. We conjecture that in the general
case of any number of agents there exist allocations guaranteeing all agents
at least $3/4$ of their maximin share. Theorems \ref{two_types_theorem} and \ref{three_types_theorem} show
the conjecture holds if agents are of three or fewer types. However, the
methods we developed so far seem too weak to prove it in full generality.

\section*{Acknowledgments}
The work of the second author was partially supported by the NSF grant
IIS-1618783.

%%%%%%%%%%%%%%%%%%%%%%%%%%%%%%%%%%%%%%%%%%%%%%%%%%%%%%%%%%%%%%%%%%%%%%%%%%%%%%%%%%%%%%%
%% The file named.bst is a bibliography style file for BibTeX 0.99c

%\newpage
%\section*{References}
\bibliographystyle{elsart-num-sort}
\bibliography{mms-biblio}

%% Appendix
\section*{Appendix - Proof of Theorem \ref{three_types_theorem}}

In this section we prove Theorem \ref{three_types_theorem}. We adhere to
the notation we used throughout the paper. In particular, we consider
a set $N=\{1,2\ldots,n\}$ of agents of three types and a cycle $C$ of $m$
goods.

In the proof, as in several other places in the paper, we work under the
assumption that agents are regular. Thus, their mms-splits consist of
non-empty bundles. Therefore, in our auxiliary results, concerning properties
of splits, we restrict attention to splits into non-empty bundles, that is,
to splits that are partitions.

We select and fix an element $a$ in $C$ and call it an \emph{anchor}. Whenever
we say that $\cX=X_1,\ldots,X_n$ is a split we we mean that it is a split of $C$ and, assume that $a\in X_1$ and
that parts $X_i$ are enumerated according to their location on the cycle as
we traverse it clockwise. Further, for every $i=1,\ldots, n$, we write
$\cX_i$ for the set of all elements in the segment of $C$ extending from the
anchor $a$ clockwise to the last element of $X_i$
(inclusive). We also assume that arithmetic expressions appearing in
indices labeling bundles in splits are evaluated modulo $n$; in particular,
for every $i=1,2,\ldots,n$, $i+n=i$. Finally, we call any set of a
split $\cX$ an {\it $X$-set}.

%A pair of splits is regular or irregular. If splits $\cX=X_1,\ldots,X_n$
%and $\cY=Y_1,\ldots,Y_n$ are irregular, then there are $i,j,k,l$ such that
%$X_i\subseteq Y_j$ and $Y_k\subseteq Y_l$.
%
%Moreover, whenever we
%consider two splits $(X_1,\ldots,X_n)$ and $(Y_1,\ldots,Y_n)$, we also assume
%that the parts are labeled so that $X_1\cap Y_1\neq \emptyset$ and $X_1$
%does not \emph{end} after $Y_1$ does (as we traverse the cycle clockwise).
%Such a labeling can always be found.
%%$X_1\cap Y_2= \emptyset$ (\textcolor{blue}{if $n\geq 3$, such a labeling
%%exists; prove it?}).
%Finally, we assume that arithmetic expressions appearing in indices
%labeling parts of splits are evaluated modulo $n$.

Assuming these conventions, let $\cX=X_1,\ldots,X_n$ and $\cY=Y_1,\ldots,Y_n$
be two splits. A bundle $X_i$
%(resp. $Y_i$)
is a {\it jump to the split $\cY$}
%(resp.  to $\cX$)
if $\cX_i \subseteq\cY_i$. In such case, we will also say that $X_i$ is a jump
to $Y_{i+1}$.
%(resp. $\cY_i\subseteq \cX_i$).
%\mc{illustrate?}
We usually omit the reference to the target split of a jump, if it is clear
from the context.
The following property follows directly from the definition.

%Moreover, a part $X_i$ (resp. $Y_i$)
%is a {\it reverse jump} if $X_i\cap Y_{i-1}=\emptyset$ (resp. $Y_i\cap X_{i-1}
%=\emptyset$). Note that for every $i$, $X_i$ ($Y_i$, respectively) is a jump
%if and only if $Y_{i+1}$ ($X_{i+1}$, respectively) is a reverse jump.
%
%Assuming these conventions, let $X_1,\ldots,X_n$ and $Y_1,\ldots,Y_n$
%be two splits.
%%and let us assume that $X_1 \subseteq Y_1$.
%A part $X_i$ (resp. $Y_i$) is a {\it jump} if $X_i\cap Y_{i+1}=\emptyset$
%(resp. $Y_i\cap X_{i+1} =\emptyset$). Moreover, a part $X_i$ (resp. $Y_i$)
%is a {\it reverse jump} if $X_i\cap Y_{i-1}=\emptyset$ (resp. $Y_i\cap X_{i-1}
%=\emptyset$). Note that for every $i$, $X_i$ ($Y_i$, respectively) is a jump
%if and only if $Y_{i+1}$ ($X_{i+1}$, respectively) is a reverse jump.

\begin{lemma}\label{lem0}
Let $\cX=X_1,\ldots,X_n$ and $\cY=Y_1,\ldots,Y_n$ be two splits.
%such that $X_1\subseteq Y_1$.
Then for every $i$, $1\leq i\leq n$, at least one of the sets $X_i$, $Y_i$ is
a jump (to the other split).
%and at least one is a reverse jump.
\end{lemma}

\begin{lemma}\label{lem1}
Let $\cX=X_1,\ldots,X_n$ and $\cY=Y_1,\ldots,Y_n$ be two splits. If for some
$i$, $1\leq i < n$, $X_i$ is a jump and $X_{i+1}$ is not a jump, then
$Y_{i+1} \subseteq X_{i+1}$ and $Y_{i+1}$ is a jump.
\end{lemma}
\begin{proof} Since $X_i$ is a jump, $\cX_i\subseteq \cY_i$. Further, since
$X_{i+1}$ is not a jump, $Y_{i+1}$ is a jump (Lemma \ref{lem0}). Thus,
$\cY_{i+1}\subseteq \cX_{i+1}$. Since $X_{i+1}=\cX_{i+1}\setminus \cX_i$
and $Y_{i+1}=\cY_{i+1}\setminus\cY_{i}$, $Y_{i+1}\subseteq X_{i+1}$
follows.
\end{proof}

Let $\cX=X_1,\ldots, X_n$ be a sequence of $n$ subsets of $C$. An
\emph{$\cX$-interval} is any sequence $X_{i}, X_{i+1}\ldots, X_{j}$ of up to
$n$ consecutive elements of $\cX$ (with $X_1$ being a successor of $X_n$).

Let $\cX=X_1,\ldots, X_n$ and $\cY=Y_1,\ldots, Y_n$ be two splits. A sequence
$\cZ=Z_1,\ldots, Z_n$ of subsets of $C$ is \emph{$\cX\cY$-useful} if
%\begin{enumerate}
%\item for every $i$, $1\leq i\leq n$, $Z_i=X_i$ or $Z_i=Y_i$
%\item
%\end{enumerate}
for every $i=1,2,\ldots,n$, at least one of the following conditions holds:
\begin{enumerate}
\item $Z_i=X_i$ and $Z_{i+1}=X_{i+1}$
\item $Z_i=Y_i$ and $Z_{i+1}=Y_{i+1}$
\item $Z_i=X_i$ and $X_i$ is a jump to $Y_{i+1}$
\item $Z_i=Y_i$ and $Y_i$ is a jump to $X_{i+1}$.
\end{enumerate}

\begin{lemma}\label{lem:useful}
Let $\cX=X_1,\ldots, X_n$ and $\cY=Y_1,\ldots, Y_n$ be splits. If a sequence
$\cZ=Z_1,\ldots, Z_n$ is $\cX\cY$-useful then the sets in $\cZ$ are pairwise
disjoint.
\end{lemma}
\begin{proof} Let us consider sets $Z_i$ and $Z_j$, where $1\leq i<j\leq n$.
From the definition of a $\cX\cY$-useful sequence, it follows that for every
$k$, $1\leq k\leq n$, $Z_k=X_k$ or $Z_k=Y_k$. Thus, $Z_i=X_i$ and $Z_j=X_j$,
$Z_i=Y_i$ and $Z_j=Y_j$, $Z_i=X_i$ and $Z_j=Y_j$, or $Z_i=Y_i$ and $Z_j=X_j$.
In the first two cases, the sets $Z_i$ and $Z_j$ are disjoint, because they
are two different members of a split. The other two cases are symmetric.
Without loss of generality, we will restrict attention to the last one:
$Z_i=Y_i$ and $Z_j=X_j$.Clearly, we may assume that $Z_i\neq X_i$ and
$Z_j\neq Y_j$ (otherwise, $Z_i$ and $Z_j$ would be members of the same
split, the case we already considered).

Let $k$ be the largest integer such that $i\leq k <j$, $Z_k=Y_k$, and $Z_{k+1}
\neq Y_{k+1}$. Since $Z_i=Y_i$, $k$ is well defined. It
follows that $Y_k$ is a jump. Thus, $\cY_k\subseteq \cX_k$. Since $k+1\leq j
\leq n$, $X_j \cap \cX_k = \emptyset$. Consequently, $X_j\cap \cY_k=\emptyset$.
Since $i\leq k$, if $i\geq 2$ then $Z_i\subseteq \cY_k$ (indeed, it follows
from the equality $Z_i= Y_i$). Thus, $Z_i$ and $Z_j$ are disjoint.

Let us assume then that $i=1$. Since the anchor $a$ belongs to $Y_1$, $Y_1$ is
the union of two segments, $A = C\setminus \cY_n$ and $B=\cY_1$. Since $X_j
\cap \cY_k = \emptyset$, $X_j \cap B=\emptyset$. Let $p$ be the smallest
integer such that $j\leq p\leq n$, $Z_p=X_p$ and $Z_{p+1}\neq X_{p+1}$. Since
$Z_j=X_j$ and $Z_1\neq X_1$ (we proved $Z_i\neq X_i$ above and here $i=1$),
$p$ is well defined. It follows that $X_p$ is a jump.
Therefore, $\cX_p\subseteq\cY_p$. Since $\cY_p\subseteq \cY_n$, $A\cap \cY_p=
\emptyset$. Since $X_j\subseteq \cX_p\subseteq \cY_p$, $X_j\cap A=\emptyset$.
Thus, $Z_j$($=X_j$) and $Z_1$($=Y_1$) are disjoint.
\end{proof}

%An $XY$-useful sequence of sets ${\cal Z}$ is {\it $(X,Y,q)$-good} if the following conditions are satisfied.
%\begin{enumerate}[(i)]
%\item there are at most two $X$-intervals and at most two $Y$-intervals in ${\cal Z}$,
%\item the total number of sets in the $Y$-intervals is $q$,
%\item at most three of the $X$- and $Y$-intervals have an odd number of sets.
%\end{enumerate}

Splits $\cX$ and $\cY$ are \emph{proper relative to each other} if
for every $X\in \cX$ and $Y \in\cY$, $X\not\subseteq Y$ and $Y\not
\subseteq X$. The following lemma provides some basic properties of
splits that are proper relative to each other (they are easy to see and
we omit the proof).

\begin{lemma}\label{lem:proper}
Let splits $\cX=X_1,\ldots,X_n$ and $\cY=Y_1,\ldots, Y_n$, with both sequences
enumerated according to our convention clockwise starting with sets
containing the anchor $a$, be proper relative to each other. Then
for every $i=1,\ldots,n$,
\begin{enumerate}
\item $X_i\subseteq Y_i\cup Y_{i+1}$ and $Y_i\subseteq X_{i-1}\cup X_i$, and
\item $X_{i-1}\cap Y_i \neq \emptyset$, $X_i\cap Y_i \neq \emptyset$, and
$X_i\cap Y_{i+1}\neq \emptyset$.
\end{enumerate}
or for every $i=1,\ldots,n$,
\begin{enumerate}
\item $X_i\subseteq Y_{i-1}\cup Y_i$ and $Y_i\subseteq X_i\cup X_{i+1}$, and
\item $X_{i}\cap Y_{i-1} \neq \emptyset$, $X_i\cap Y_i \neq \emptyset$, and
$X_{i+1}\cap Y_i\neq \emptyset$.
\end{enumerate}
\end{lemma}

\nop{
A pair of splits $(\cX,\cY)$, where $\cX=X_1,\ldots,X_n$ and $\cY=Y_1,\ldots,
Y_n$ (with both sequences enumerated according to our convention, clockwise
starting with sets containing the anchor $a$), is \emph{proper}
if for every
$i,j$, $1\leq i,j\leq n$, $X_i\not \subseteq Y_j$ and $Y_j\not\subseteq X_i$.
Because of the way elements in $\cX$ and $\cY$ are enumerated, it follows that
if a pair of splits $(\cX,\cY)$ is proper, then for every $i=1,\ldots,n$,
\begin{equation}
\label{eq:reg-1}
X_i\subseteq Y_i\cup Y_{i+1}\ \mbox{and}\ Y_i\subseteq X_{i-1}\cup X_i.
\end{equation}
or for every $i=1,\ldots,n$,
\begin{equation}
\label{eq:reg-2}
X_i\subseteq Y_{i-1}\cup Y_i\ \mbox{and}\ Y_i\subseteq X_i\cup X_{i+1}.
\end{equation}
Moreover, in each case, each $X$-set has a non-empty intersection with exactly
two $Y$-sets (which are necessarily consecutive) and each $Y$-set has a
non-empty intersection with exactly two $Y$-sets (which are necessarily
consecutive). We say that a proper pair $(\cX,\cY)$ of splits, where
$\cX=X_1,\ldots,X_n$ and $\cY=Y_1,\ldots, Y_n$, is \emph{regular} if
(\ref{eq:reg-1}) holds for every $i=1,\ldots,n$. It is clear that if a pair
$(\cX,\cY)$ of splits is proper then either $(\cX,\cY)$ or $(\cY,\cX)$ is
regular.}

From now on we consider splits that are mms-splits for a set $N$ of regular
agents. In particular, all agents are proportional and for every agent
the total utility of all goods in $C$ is $n$. We also use the following
terminology. A subset $X$ of $C$ that induces in $C$ a connected subgraph
is \emph{acceptable} to an agent $x$ if the total value of the goods in $X$
to $x$ is at least $\frac{3}{4}$.

\begin{lemma}\label{every_second}
Let $x$ and $y$ be two regular agents from $N$. Let $\cX$ and $\cY$ be
mms-splits of a cycle for agents $x$ and $y$, respectively, and let $\cX$
and $\cY$ be proper relative to each other. If no set $X\cap Y$, where $X\in
\cX$ and $Y\in \cY$, is acceptable to $y$, then at least one of each pair of
consecutive sets of $\cX$ is acceptable to $y$.
\end{lemma}
\begin{proof} To prove the lemma let us consider any two consecutive sets
in $\cX$, say $X$ and $X'$, with $X'$ directly following $X$ in $\cX$.
Since $\cX$ and $\cY$ are proper relative to each other, Lemma \ref{lem:proper}
implies that there are three
consecutive sets $Y,Y',Y''\in \cY$ such that $X\cup X'\subseteq Y\cup Y'
\cup Y''$. Let $X_{-1}$ and $X'_{+1}$ be the predecessor of $X$ in $\cX$
and the successor of $X'$ in $\cX$, respectively. By our assumption,
$u_y(X_{-1}\cap Y)<\frac{3}{4}$ and $u_y(X'_{+1}\cap Y'')<\frac{3}{4}$,
where $u_y$ is the utility function for agent $y$. It follows that
\begin{align*}
3 &= u_y(Y\cup Y'\cup Y'')\\
  &= u_y(X_{-1}\cap Y)+u_y(X)+ u_y(X')+u_y(X'_{+1}\cap Y'')\\
  &< 2\cdot\frac{3}{4} + u_y(X)+u_y(X').
\end{align*}
Thus, $u_y(X)\geq \frac{3}{4}$ or $u_y(X')\geq\frac{3}{4}$, that is, at least
one of $X$ and $X'$ is acceptable to $y$.
\end{proof}

\nop{
\begin{lemma}\label{every_second}
Let $x$ and $y$ be two regular agents from $N$. Let $\cX=X_1,\ldots,X_n$
and $\cY= Y_1,\ldots,Y_n$ be mms-splits of a cycle for agents $x$ and $y$,
respectively, and let $(\cX,\cY)$ be regular. Then, at least one of each
two consecutive sets of $\cX$ is aceptable to agent $y$, or
for some $i$, $1\leq i\leq n$, at least one of the sets $X_{i-1}\cap Y_{i}$,
$X_{i+2}\cap Y_{i+2}$ is acceptable to $y$.
\end{lemma}
\begin{proof} To prove the lemma let us suppose that for some $i$, $1\leq
i\leq n$, neither $X_i$ nor $X_{i+1}$ is acceptable to $y$. By (\ref{eq:reg-1}),
$X_i\cup X_{i+1}\subseteq Y_i\cup Y_{i+1} \cup Y_{i+2}$. Thus,
\begin{align*}
3&= u_y(Y_i\cup Y_{i+1} \cup Y_{i+2})\\
 &= u_y(X_{i-1}\cap Y_i)+u_y(X_i)+ u_y(X_{i+1})+u_y(X_{i+2}\cap Y_{i+2})\\
 &< 2\cdot\frac{3}{4} + u_y(X_{i-1}\cap Y_i)+u_y(X_{i+2}\cap Y_{i+2}),
\end{align*}
where $u_y$ is the utility function for agent $y$. It follows that
$u_y(X_{i-1}\cap Y_i)\geq \frac{3}{4}$ or $u_y(X_{i+2}\cap Y_{i+2})
\geq\frac{3}{4}$.
\end{proof}
}

In the next lemma and in the proof of Theorem \ref{three_types_theorem},
we use the following notation. We write $n_1,n_2,n_3$ for the numbers of agents
of types $i_1,i_2,i_3$, respectively. In particular, $n=n_1+n_2+n_3$. Further,
we assume without loss of generality that $n_1\geq n_2\geq n_3\geq 1$.
%and that the total value of all nodes on $C$ is $n$ for each agent.

\begin{lemma}\label{n_1=n_2}
Let $\cA= A_1,\ldots,A_n$ (resp. $\cB=B_1,\ldots,B_n$ and $\cC=C_1,\ldots,C_n$)
be mms-splits of the cycle $C$ for agents of type $i_1$ (resp. $i_2$ and $i_3$).
If $n_1=n_2$, $n_3\geq 1$, no set $A_i\cap B_j$
is acceptable to any agent, and some $C$-set is contained in some $A$-set or
$B$-set, then there is a $\frac{3}{4}$-sufficient allocation for a cycle and
any number of agents of three types.
\end{lemma}
\begin{proof} %By Lemma \ref{every_second}, at least one of any two
%consecutive $A$-sets is acceptable for agents $i_2$ and at least one of any
%two consecutive $B$-sets is acceptable for agents $i_1$.
%
%\smallskip
Since no set $A_i\cap B_j$ is acceptable to any agent, we have that for
every $i$ and $j$, where $1\leq i,j\leq n$, $A_i\not\subseteq B_j$
and $B_i\not\subseteq A_j$. Thus, the splits $\cA$ and $\cB$ are
proper relative to each other.

Our assumptions do not distinguish between the splits $\cA$ and
$\cB$. Therefore, without loss of generality, we may assume that a $C$-set,
say $C_k$, is contained in an $A$-set, say $A_i$. Let $B_j$ and $B_{j+1}$ be
two $B$-sets that intersect $A_i$. It follows that $C_k\subseteq B_j\cup
B_{j+1}$. If $C_k\subseteq B_{j+1}$, then $C_k \subseteq A_i\cap B_{j+1}$.
But then $A_i\cap B_{j+1}$ would be acceptable to agents of type $i_3$ (and
there is at least one such agent). This is a contradiction. Thus, $C_k\cap
B_j\neq \emptyset$. Consequently, $C_k\cap B_j\cap A_i\neq \emptyset$.
Without loss of generality, we may assume that the anchor $a$ belongs to
$C_k\cap B_j\cap A_i$. Assuming the sets in the splits $\cA$, $\cB$ and $\cC$
are enumerated with respect to that anchor, we have $k=j=i=1$, $C_1 \subseteq
A_1$ and $B_1\subseteq A_n\cup A_1$. The latter implies (we recall that $\cA$
and $\cB$ are proper relative to each other and so, Lemma \ref{lem:proper}
applies) that for every $i=1,\ldots,n$,
\begin{equation*}
A_i\subseteq B_i\cup B_{i+1}\ \mbox{and}\ B_i\subseteq A_{i-1}\cup A_i.
\end{equation*}

In what follows we will be applying Lemmas \ref{lem0} - \ref{every_second}.
We can do so, as we follow here the same convention for enumerating sets
in splits with respect to the same fixed anchor.

\smallskip
\noindent
{\bf Case 1} Some two consecutive $A$-sets are jumps to $\cC$.

%\smallskip
%Let us assume that these two consecutive $A$-sets are jumps. The other case
%is symmetric.

\smallskip
\noindent
If all $A$-sets are jumps to $\cC$, then the sequence
\begin{equation*}%\label{sequence0}
\cZ=C_1,A_2,\ldots,A_{2n_1+1},C_{2n_1+2},\ldots,C_n
\end{equation*}
is $\cA\cC$-useful. Indeed, $C_1$ is a jump as it is a subset of $A_1$
(thus, $\cC_1\subseteq \cA_1$, in the notation we introduced to
define jumps),
and $A_{2n_1+1}$ is a jump by assumption. By Lemma
\ref{lem:useful}, all sets in $\cZ$ are pairwise disjoint.

By assumption, no set $A_i\cap B_j$ is acceptable to any agent
of type $i_2$. Thus, by Lemma \ref{every_second}, some $n_1$($=n_2$) of
different sets in the interval $A_2,\ldots,A_{2n_1+1}$ in $\cZ$ are
acceptable to agents of type $i_2$. We assign
them to those agents. We assign the remaining $n_1$ sets in
$A_2,\ldots,A_{2n_1+1}$
to agents of type $i_1$. Finally, we assign the sets $C_j$, $j=1, 2n_1+2,
\ldots, n$, to agents of type $i_3$. In this way all agents are allocated
sets that are acceptable to them. If there are any unallocated
goods ($\cZ$ does not have to cover all goods in $C$), they can be distributed
to agents so that all agents receive bundles inducing in $C$ a path.
This, yields an assignment that is a $\frac{3}{4}$-sufficient allocation.

So, assume that some of $A$-sets are not jumps. Since some two consecutive
$A$-sets are jumps, there is an integer $i$ such that $A_{i-1}$ and $A_i$
are jumps and $A_{i+1}$ is not.
By Lemma \ref{lem1}, $C_{i+1}$ is a jump and $C_{i+1}\subseteq A_{i+1}$.
The latter property implies that we may relabel sets $A_i$ and $C_i$ so that
$A_{n-1}$ and $A_n$ are jumps and $A_1$ is not. Indeed, under this relabeling
we have $C_1\subseteq A_1$, which we assumed to hold when we started
the proof (of course, we also need to relabel sets $B_i$
correspondingly).

We now consider the case when $C_{n_3}$ or $C_{n_3-1}$ is a jump.
In the first case, we set
\begin{equation*}%\label{sequence1}
\cZ=C_{1},\ldots,C_{n_3},A_{n_3+1},\ldots,A_{n}
\end{equation*}
and, in the second case,
\begin{equation*}%\label{sequence1.5}
\cZ=C_{n},C_{1},\ldots,C_{n_3-1},A_{n_3},\ldots,A_{n-1}.
\end{equation*}
Since $A_n$ and $A_{n-1}$ are jumps, in either case $\cZ$ is $\cA\cC$-useful.
Thus, by Lemma \ref{lem:useful}, its elements are pairwise disjoint.
Reasoning as above, we distribute the sets $A_{n_3+1},\ldots,A_{n}$
(respectively, $A_{n_3},\ldots,A_{n-1}$) among $2n_1$ agents of types
$i_1$ and $i_2$ (since $n=2n_1+n_3$, in each case there are $2n_1$
sets in that sequence), and then allocate the remaining $n_3$ sets, all
of them $C$-sets, to agents of type $i_3$. Clearly, the resulting
allocation (after possibly attaching goods not ``covered'' by $\cZ$ to
appropriate sets in $\cZ$) is $\frac{3}{4}$-sufficient.

Next, let us assume that neither $C_{n_3}$ nor $C_{n_3-1}$ is a jump.
Then let $k$ be the largest integer such that $1\leq k < n_3-1$ and $C_k$
is a jump (such $k$ exists because $C_1$ is a jump). Moreover, let $\ell$
be the smallest integer $j$ such that $n_3 < j \leq n$ and $C_{j}$ is a jump,
if such $j$ exists, or $\ell=n$, otherwise. Observe that the sequence
\begin{equation*}%\label{sequence2}
\cZ=C_{1},\ldots,C_{k},A_{k+1},\ldots,A_{\ell-(n_3-k)},C_{\ell-(n_3-k)+1},\ldots,C_{\ell},A_{\ell+1}\ldots,A_{n}
\end{equation*}
is $\cA\cC$-useful. Indeed, for $k<j<\ell$ the sets $C_j$ are not jumps.
By Lemma \ref{lem0}, the sets $A_j$, where $k<j<\ell$ are jumps. Since
$k<\ell -(n_3-k) < \ell$, $A_{\ell -(n_3-k)}$ is a jump. Moreover, $C_k$
is a jump based by our choice of $k$, $C_\ell$ is a jump by our choice of
$\ell$, if $\ell$ is not assigned to $n$ by default, and $A_n$ is a jump by
assumption. By Lemma \ref{lem:useful}, all sets in $\cZ$ are
pairwise disjoint.

Clearly, if $\ell$ and $n$$(=2n_1+n_3$) have the same parity, then the
numbers of sets in the intervals $A_{k+1},\ldots,A_{\ell-(n_3-k)}$ and
$A_{\ell+1},\ldots, A_n$ are even (the latter is present and needs to be considered
only if $\ell<n$). Moreover, the
total number of sets in these two intervals is $n-n_3=2n_1=2n_2$.
By Lemma \ref{every_second},
we can select among these $A$-sets in $\cZ$ $n_1$($=n_2$) sets that are
acceptable to agents of type $i_2$. We allocate these sets to those agents.
We then allocate the remaining $n_1$ of these $A$-sets to the agents of
type $i_1$. Finally, agents of type $i_3$ receive the $n_3$ $C$-sets of
$\cZ$. Extending $\cZ$ to an allocation of $C$, yields an allocation
that is $\frac{3}{4}$-sufficient.

So, assume now that $\ell$ and $n$ have different parity. In this case, we
define $\ell'$ to be the smallest integer $j$ such that $n_3 < j \leq n-1$
and $C_{j}$ is a jump, if such $j$ exists, or $\ell'=n-1$ otherwise.
Obviously, $\ell'=\ell$ or $\ell'=n-1$, so in both cases $\ell'$ and $n$ have different parity. Since
$A_{n-1}$ is a jump, the same reasoning as above yields that the sequence
\begin{align*}
\cZ&=\\
&C_n,C_{1},\ldots,C_{k},A_{k+1},\ldots,A_{\ell'-(n_3-k)+1},C_{\ell'-(n_3-k)+2},\ldots,C_{\ell'},A_{\ell'+1}\ldots,A_{n-1}
\end{align*}
is $\cA\cC$-useful. Thus, its elements are pairwise disjoint. Moreover, the
$A$-intervals $A_{k+1},\ldots,A_{\ell'-(n_3-k)+1}$ and $A_{\ell'+1},
\ldots, A_{n-1}$ (if $\ell' < n-1$) consist of even numbers of sets
and the total number of sets in these two intervals is $2n_1=2n_2$.
Therefore, we can define $\frac{3}{4}$-sufficient allocation in a similar way
as when $\ell$ and $n$ have the same parity.

\smallskip
\noindent
{\bf Case 2} Some two consecutive $C$-sets are jumps to $\cA$.

\smallskip
\noindent
This case is symmetric to the previous one if we enumerate sets
counter\-clockwise rather than clockwise. Indeed, if a $C$-set $X$ is a jump
to an $A$-set $Y$ under the clockwise enumeration, then $Y$ is a jump to $X$
under the counterclockwise enumeration.

\smallskip
\noindent
{\bf Case 3} No two consecutive $A$-sets are jumps to $\cC$, and no two
consecutive $C$-sets are jumps to $\cA$.

\smallskip
\noindent
Since $C_1\subseteq A_1$, $C_1$
is a jump. It follows that $C_n\cap A_1\neq\emptyset$ as otherwise, $C_n$
would be a jump. In particular, $C_n\not\subseteq A_n$. Applying Lemma
\ref{lem1} repeatedly, yields that for all odd $i$, $1\leq i\leq n$, $C_i
\subseteq A_i$ and for every even $i$, $1\leq i \leq n$, $A_i\subseteq C_i$.
Further, since $C_n\not\subseteq A_n$, it follows that $n$ is even.

We will now show that the splits $\cB$ and $\cC$ are proper
relative to each other. To this end, we need to show that for every $i$ and
$j$, $1\leq i,j\leq n$, $C_i\not\subseteq B_j$ and $B_j\not\subseteq C_i$.
First, let us assume that $C_i\subseteq B_j$, for some $i$ and $j$, $1\leq i,
j\leq n$. If $A_i \subseteq C_i$, then $A_i \subseteq B_j$, a contradiction
as the splits $\cA$ and $\cB$ are proper relative to each other. If $C_i
\subseteq A_i$, then $C_i \subseteq A_i \cap B_j$. Therefore, $A_i\cap B_j$
is acceptable to agents of type $i_3$, a contradiction with the assumptions
of the lemma. Thus, for every $i,j$ such that $1\leq i,j \leq n$, $C_i\not
\subseteq B_j$.

Next, let us assume that $B_j\subseteq C_i$, for some $i$ and $j$,
$1\leq i,j\leq n$. If $i$ is odd, then $C_i\subseteq A_i$. It follows
that $B_j\subseteq A_i$, a contradiction with the fact that $\cA$ and $\cB$ are proper relative to each other.
Thus, $i$ is even,
%\sout{$A_i\subseteq C_i$ and $A_k\cap C_i=\emptyset$ for $k\neq i-1,i,i+1$.
%Further, we recall that we have $B_j\subseteq A_{j-1} \cup A_j$. Thus, $B_j
%\subseteq (C_i\cap A_{j-1}) \cup (C_i\cap A_j)$. If $C_i\cap A_{j-1}=
%\emptyset$, $B_j\subseteq \sout{A_{j-1}}\blue{A_j}$,
%a contradiction. Thus, $C_i\cap A_{j-1} \neq\emptyset$. Similarly, $C_i\cap
%A_{j} \neq\emptyset$, too. Since $A_k \cap C_i=\emptyset$ for $k\neq i-1,i,
%i+1$, it follows that $j=i$ or $j=i+1$.}
so $B_j\subseteq C_i\subseteq C_{i-1}
\cup C_i\cup C_{i+1}\subseteq A_{i-1}\cup A_i\cup A_{i+1}$. We recall that we have
$B_j\subseteq A_{j-1}\cup A_j$. Hence, $j=i$ or $j=i+1$.

Let us assume that $j=i$. Since $i$ is even, $C_{i-1}\subseteq A_{i-1}\subseteq B_{i-1}\cup B_i$.
Moreover, $B_i\subseteq C_i$, so $B_i\cap C_{i-1}=\emptyset$. Thus,
$C_{i-1}\subseteq B_{i-1}$, a contradiction with what we already proved
above. It follows that we must have $j=i+1$. In this case we similarly get, $C_{i+1}\subseteq A_{i+1}\subseteq B_{i+1}\cup B_{i+2}$ and $B_{i+1} \subseteq
C_i$ (so, $B_{i+1}\cap C_{i+1}=\emptyset$). Thus, $C_{i+1}\subseteq B_{i+2}$,
a contradiction as in the case $j=i$. This completes the proof that for every
$i,j$ such that $1\leq i,j \leq n$, $B_j\not\subseteq C_i$, and show that the
splits $\cB$ and $\cC$ are proper relative to each other.

Let us observe that $B_1\subseteq A_n\cup A_1$, $A_n\subseteq C_n$, $B_1\cap C_1\not=\emptyset$
and $C_1\not\subseteq B_1$. It follows that $B_1\subseteq
C_n \cup C_1$. Since the splits $\cB$ and $\cC$ are proper relative to each
other, Lemma \ref{lem:proper} implies that for every $i$, $1\leq i
\leq n$, $B_i\subseteq C_{i-1}\cup C_i$.

\nop{We will now show that the pair $(\cB,\cC)$ of splits is regular. Indeed,
otherwise $C_i\subseteq B_j$, for some $i$ and $j$, $1\leq i,j\leq n$. If $A_i
\subseteq C_i$, then $A_i \subseteq B_j$, a contradiction with the regularity
of $(\cA,\cB)$. If $C_i\subseteq A_i$, then $C_i \subseteq A_i
\cap B_j$. Therefore, $A_i\cap B_j$ is acceptable to agents of type $i_3$,
a contradiction with the assumptions of the lemma.

Because $(\cB,\cC)$ is regular, we may assume without loss of generality
that for every $i$, $1\leq i\leq n$, $B_i\subseteq C_i\cup C_{i+1}$.}

Let us suppose that some set $B_i\cap C_j$, where $1\leq i,j\leq n$, is
acceptable for agents of type $i_2$. Then, there is $i$, $1\leq i\leq n$,
such that $B_i\cap C_{i-1}$ or $B_i\cap C_i$ has this property.

Let us assume the former. We define a $\frac{3}{4}$-sufficient allocation as
follows. We consider the sequence $C_{i},C_{i+1},\ldots,C_{i+n_3}$ of $n_3+1$
$C$-sets. One of the sets $C_{i},C_{i+1}$ contains an $A$-set and that set
is assigned to an agent of type $i_1$. The remaining $n_3$ sets of that sequence
go to agents of type $i_3$. The sequence $B_{i-1}, B_{i-2},
\ldots, B_{i-(2n_1-2)}$ contains at least $n_1-1$ sets acceptable for agents
of type $i_1$ (Lemma \ref{every_second}). We select some $n_1-1$ of such
sets and assign them to agents of type $i_1$. We allocate the remaining
$n_1-1$ sets
from this sequence and the set $B_i\cap C_{i-1}$ to agents of type $i_2$. In
this way, all agents are allocated sets that are acceptable to them. Moreover,
these sets are pairwise disjoint. Indeed, the sets $C_{i+n_3}$ and
$B_{i-(2n_1-2)}$ do not overlap because $i-(2n_1-2) = i+n-2n_1+2 =i+n_3+2$
(we recall that the arithmetic of indices is modulo $n$), and $B_{i+n_3+2}
\subseteq C_{i+n_3+1}\cup C_{i+n_3+2}$. Thus, the assignment is a
$\frac{3}{4}$-sufficient allocation.

The latter case, when $B_i\cap C_{i}$ is acceptable to an agent of type
$i_2$ can be handled similarly by considering sequences $C_{i-1},C_{i-2},
\ldots, C_{i-(n_3+1)}$ and $B_{i+1},B_{i+2},$ $\ldots, B_{i+2n_1-2}$.

Assume now that no set $B_i\cap C_j$ is acceptable for agents of type $i_2$.
By Lemma \ref{every_second} one in every two consecutive sets of the mms-split
$C_1,\ldots,C_n$ is acceptable for agents of type $i_2$. Clearly, $A_n$ is a
jump to the split $\cC$ because $C_1\subseteq A_1$. Moreover, $C_{n_3}\subseteq A_{n_3}$ or
$C_{n_3+1}\subseteq A_{n_3+1}$. Therefore, $C_{n_3+j}$ is a jump for some
$j\in\{0,1\}$. It follows that the sequence
\[
\cZ=C_1,\ldots,C_{n_3+j},A_{n_3+j+1},\ldots,A_n
\]
is $\cA\cC$-useful. By the comment above, at least $j$ sets in $C_1,\ldots,
C_{n_3+j}$ are acceptable to agents of type $i_2$ (since $n_3\geq 1)$. We
assign $j$ of these
sets to agents of type $i_2$. The remaining $n_3$ sets of $C_1,\ldots,C_{n_3+j}$
are assigned to agents of type $i_3$. Next, the sets $A_{n_3+j+1},\ldots,A_n$
are distributed among agents of types $i_1$ and $i_2$. Agents of type $i_2$
take $n_2-j$ sets of this interval (there are this many sets acceptable to
agents of type $i_2$ because $\left\lfloor\frac{n-n_3-j}{2}\right\rfloor=
\left\lfloor\frac{2n_2-j}{2}\right\rfloor\geq n_2-j$). The agents of type
$i_1$ receive the remaining $n_1$ sets of the $A$-interval. By construction,
this allocation is $\frac{3}{4}$-sufficient.
\end{proof}

\noindent
%\begin{theorem}\label{tthree_types_theorem}
\textbf{Theorem 4.9.}
{\it Let $C$ be a cycle of goods and $N$ a set of agents of at most three types.
Then, a $\frac{3}{4}$-sufficient allocation of goods on $C$ to agents in $N$
exists.}

\begin{proof}
We adhere to the notation from Lemma \ref{n_1=n_2}. Further, as there,
we assume that $n_3\geq 1$ (otherwise, we are in the case of agents of two
types settled by Theorem \ref{two_types_theorem}). We now consider
several cases. In each case we assume that none of the cases considered
earlier applies.

\smallskip
\noindent
{\bf Case 1.} Some set $A_i\cap B_j$ is acceptable to some agent.

\noindent
\smallskip
In this case, the existence of a $\frac{3}{4}$-sufficient allocation follows
directly from Lemma \ref{lem_t_types}.
%Let $P'$ be a minimal subpath of the path $A_i\cap B_j$ which is acceptable
%for some agent. Let $P=C-P'$. Since $P'\subseteq A_i\cap B_j\subseteq A_i$,
%the path $P$ contains all parts of the mms-split $A_1,A_2,\ldots,A_n$ of $C$
%except $A_i$. Therefore there is an mms-split of $P$ for agent $i_1$.
%Similarly, agent $i_2$ has an mms-split of $P$. Let $S=\{ i_1,i_2\}$, so
%$R=\{ i_3\}$.
%
%Since $n_1\geq n_2\geq n_3$,
%\[
%\frac{n}{3}=\frac{n_1+n_2+n_3}{3}\geq n_3=|R|,
%\]
%so we are done by Theorem \ref{algorithm}.

From now on, we assume Case 1 does not apply, that is, no set
$A_i\cap B_j$
is acceptable to any agent. It follows that for every $1\leq i,j\leq n$,
$A_i\not\subseteq B_j$ and $B_j\not\subseteq A_i$. Thus, the splits $\cA$
and $\cB$ are proper relative to each other. This means, in particular, that
from now on we may apply Lemma \ref{every_second} to splits $\cA$ and $\cB$,
as well as to $\cB$ and $\cA$.

\smallskip
\noindent
{\bf Case 2.} Some $C$-set is contained in some $A$-set.

\smallskip
\noindent
Given our comments above, Lemma \ref{n_1=n_2} applies and we can assume that
$n_1>n_2\geq n_3$. We can also assume without loss of generality that
$C_{1}\subseteq A_{1}$. Then, the sets $A_n$ and $C_{1}$ are jumps.

If the set $C_{n_3}$ is a jump too, then the %cyclic
sequence
\[
\cZ=C_{1},\ldots,C_{n_3},A_{n_3+1},\ldots,A_n
\]
is $\cA\cC$-useful. By Lemma \ref{every_second} (which also applies now),
at least $\lfloor\frac{n_1+n_2}{2}\rfloor\geq n_2$ of the sets
that form the sequence $A_{n_3+1},\ldots,A_n$ are acceptable for agents of
type $i_2$. We assign $n_2$ of them to these agents. We assign the remaining
$n_1$ sets of this interval to agents of type $i_1$. Agents of type $i_3$
receive the sets $C_{1},\ldots,C_{n_3}$. In this way, we construct a
$\frac{3}{4}$-sufficient allocation.

If $C_{n_3}$ is not a jump, then let $k$ be the largest integer such that
\green{$1\leq k < n_3$} and $C_{k}$ is a jump (such $k$ exists because $C_1$ is a
jump). Moreover, let $\ell$ be the smallest integer $j$ such that $n_3 < j
\leq n$ and $C_{j}$ is a jump, if such $j$ exists, or $\ell=n$ otherwise.
We define
\begin{equation*}%\label{sequence2}
\cZ=C_{1},\ldots,C_{k},A_{k+1},\ldots,A_{\ell-(n_3-k)},C_{\ell-(n_3-k)+1},\ldots,C_{\ell},A_{\ell+1}\ldots,A_{n}
\end{equation*}
and observe that it is $\cA\cC$-useful. Indeed, $C_k$ is a jump by the construction. Further, $k<\ell-(n_3-k)<\ell$ and for $k<j<\ell$ the sets $C_j$ are not
jumps. Thus, $A_{\ell-(n_3-k)}$ is a jump by Lemma \ref{lem0}. Finally, if
$\ell<n$, $C_\ell$ is a jump by the choice of $\ell$ and $A_n$ is a jump by
assumption.

By Lemma \ref{every_second}, there are at least
$\left\lfloor\frac{\ell-n_3}{2}\right\rfloor+
\left\lfloor\frac{n-\ell}{2}\right\rfloor
\geq \frac{\ell-n_3-1}{2}+\frac{n-\ell-1}{2}
=\frac{n_1+n_2}{2}-1>n_2-1$
sets in the intervals $A_{k+1},\ldots,A_{\ell-(n_3-k)}$ and
$A_{\ell+1}\ldots,A_{n}$ (the latter present only if $\ell<n$) that
are acceptable to agents of type $i_2$. We assign $n_2$ of them to these
agents. The remaining sets of these intervals (there are $n_1$ of them)
are assigned to agents of type $i_1$. Finally, agents of type $i_3$ receive
the sets in the intervals $C_{1},\ldots,C_{k}$ and $C_{\ell-(n_3-k)+1},
\ldots,C_{\ell}$. This yields a $\frac{3}{4}$-sufficient allocation.

\smallskip
\noindent
{\bf Case 3.} Some $A$-set is contained in some $C$-set.

\smallskip
\noindent
Without loss of generality we may assume that $A_1\subseteq C_1$. Since no $C$-set is
contained in an $A$-set (that possibility is excluded by Case 2), a simple
inductive argument shows that for every $i=2,\ldots, n$, $A_i$ ends
\emph{strictly} before $C_i$ does. In particular, $A_n$ ends strictly before
$C_n$ does. This implies that $A_1\cap C_n\neq \emptyset$, a contradiction
with $A_1\subseteq C_1$ (and the fact that different sets in a split are
disjoint).

Given Cases 2 and 3, from now on we may assume that the splits
$\cA$ and $\cC$ are proper relative to each other.

\smallskip
\noindent
{\bf Case 4.} Some set $A_i\cap C_j$ is acceptable to agents of type $i_3$.

\smallskip
\noindent
As we noted, we may assume that the splits $\cA$ and $\cC$ proper are relative to
each other. Thus, we can relabel the sets in the split $\cC$ if necessary
so that for all $i$, $A_i\subseteq C_i\cup C_{i+1}$ and, consequently, also
$C_i\subseteq A_{i-1}\cup A_i$.

\smallskip
\noindent
We can assume without loss of generality that $A_1\cap C_1$ or $A_{n}\cap
C_1$ is acceptable to agents of type $i_3$. Let us assume the former. We
allocate $A_1\cap C_1$ and the sets $C_2,C_3,\ldots,C_{n_3}$ to agents of
type $i_3$. Clearly, all these sets are included in $A_1\cup\ldots\cup
A_{n_3}$. The sequence $A_{n_3+1},\ldots, A_n$ of the remaining $A$-sets
contains $n_1+n_2$ sets. By Lemma \ref{every_second}, at least $\left\lfloor
\frac{n_1+n_2}{2} \right\rfloor\geq n_2$ of them are acceptable for agents
of type $i_2$. We allocate any $n_2$ of those sets to agents of type $i_2$
and the remaining $n_1$ of them to the agents of type $i_1$. The allocation
we defined in this way extends to a $\frac{3}{4}$-sufficient allocation for
$C$.

The other case, when $A_{n}\cap C_1$ is acceptable to agents of
type $i_3$, can be handled similarly by considering sequences $C_n,C_{n-1},
\ldots,C_{n-n_3+2}$ and $A_1,A_2,\ldots,A_{n-n_3}$.

\nop{
%Could replace the magenta text above
If $A_{n}\cap C_1$ is acceptable to agents of type $i_3$, we reason
similarly. We allocate $A_{n}\cap C_1$, $C_n, C_{n-1},\ldots, C_{n-(n_3-2)}$
to agents of type $i_3$. Since these sets are contained in $A_n\cup A_{n-1}
\cup A_{n-n_3+1}$, the sets $A_1,\ldots, A_{n_1+n_2}$ are still unallocated
(indeed, observe that $n-n_3+1=n_1+n_2+1$). Thus, they can be allocated among
the agents of type $i_1$ and $i_2$ so that each agent receives a set that is
acceptable to her (the same argument as above applies). Clearly, this
allocation can be modified to yield a $\frac{3}{4}$-sufficient allocation
for $C$.}

Excluding Cases 1-4 means in particular that $\cA$ and $\cC$ are proper relative
to each other and that no set $A_i\cap C_j$ is acceptable to agents of type $i_3$.
Therefore, from now on we may apply Lemma \ref{every_second} also to the
splits $\cA$ and $\cC$.

\smallskip
\noindent
{\bf Case 5.} $n_1\geq\left\lfloor\frac{n}{2}\right\rfloor$.

\smallskip
\noindent
Let $A_i$ be any $A$-set acceptable to an agent of type $i_3$. Such
sets exist by Lemma \ref{every_second}. Without loss of generality we may
assume that $A_n$ is acceptable to an agent of type $i_3$. By Lemma
\ref{every_second}, the sequence $A_1, A_2,\ldots, A_{2n_2}$ contains $n_2$
sets that are acceptable
to agents of type $i_2$. We allocate these sets to them. Next, we consider
the sequence $A_{2n_2+1},A_{2n_2+2},\ldots,A_{n-1}$ of $n-1-2n_2$ sets
that follow $A_{2n_2}$. Since $n_1\geq \lfloor \frac{n}{2}\rfloor=\lceil
\frac{n-1}{2}\rceil$, we have $\lfloor\frac{n-1}{2}\rfloor \geq n_2+n_3-1$. Thus,
$\lfloor\frac{n-1-2n_2}{2}\rfloor\geq n_3-1$. By Lemma \ref{every_second},
it follows that some $n_3-1$ sets in the sequence $A_{2n_2+1},A_{2n_2+2},
\ldots,A_{n-1}$ are
acceptable to agents of type $i_3$. We allocate these sets and the set $A_n$
to agents of type $i_3$. We allocate the remaining $n_1$ sets in $A_1,\ldots,
A_n$ to agents of type $i_1$. This defines a $\frac{3}{4}$-sufficient
allocation for $C$.

\smallskip
\noindent
{\bf Case 6.} Some set $A_i\cap C_j$ is acceptable to agents of type $i_1$ and $n_1>n_2$.

\smallskip
\noindent
By the same argument as in Case 4, we may assume that $C_i\subseteq A_{i-1}\cup A_i$ for all $i$ and that
$A_1\cap C_1$ or
%$A_1\cap C_n$
$A_n\cap C_1$
is acceptable to agents of type $i_1$. If $A_1\cap C_1$ is acceptable to
agents of type $i_1$, we proceed as follows. We allocate the sets $C_2,C_3,
\ldots,C_{n_3+1}$ to agents of type $i_3$. Next, we note that every other
set in the sequence $A_{n_3+2}, A_{n_3+3},\ldots,$ $A_n$
%A_n,A_{n-1},\ldots,A_{n_3+2}$
is acceptable to agents of type $i_2$. Since the sequence contains $n-n_3-1
= n_1+n_2-1\geq 2n_2$ sets (recall that we assume here that $n_1>n_2$), some
$n_2$ of the sets in the sequence are acceptable to agents of type $i_2$. We
allocate these $n_2$ sets to agents of type $i_2$. We allocate the remaining
$n_1-1$ sets in that sequence and the set $A_1\cap C_1$ to agents of type
$i_1$. In this way, all agents get sets that are acceptable to them. Further, we
note that $C_2\cup\ldots \cup C_{n_3+1} \subseteq A_1\cup\ldots\cup A_{n_3+1}$.
It follows that all sets used in the allocation are pairwise disjoint and so
give rise to a $\frac{3}{4}$-sufficient allocation for $C$.

The case when $A_1\cap C_n$ is acceptable to an agent of type $i_1$
follows from the one we just considered by the same argument as that used in
Case 4.

\nop{
The case when $A_n\cap C_1$ is acceptable to an agent of type $i_1$ is
similar. We allocate $n_3$ sets $C_{n-n_3+1}, C_{n-n_3+2}, \ldots, C_n$
to agents of type $i_3$. For the same reason as before, the sequence $A_1,
\ldots, A_{n_1+n_2-1}$ contains at least $n_2$ sets that are acceptable to
agents of type $i_2$. We allocate some $n_2$ of them to these agents. We
allocate the remaining $n_1-1$ of these sets and the set $A_n\cap C_1$ to
agents of type $i_1$. In this way, all agents get sets that are acceptable
to them. Since $A_1\cup\ldots\cup A_{n_1+n_2-1} \subseteq C_1\cup\ldots\cup
C_{n_1+n+2}$ and $n-n_3+1 = n_1+n_2+1$, the sets we allocated are pairwise
disjoint. Thus, they yield a $\frac{3}{4}$-sufficient allocation for $C$.}

\smallskip
\noindent
{\bf Case 7.} Some $C$-set is contained in some $B$-set.

\smallskip
\noindent
By Lemma \ref{n_1=n_2}, we can assume that $n_1>n_2$ (because Case 1 is
excluded, $\cA$ and $\cB$ are proper relative to each other and no set
$A_i\cap B_j$ is acceptable to any agent). We now recall that splits $\cC$
and $\cA$ are proper relative to each other. Moreover, no set $A_i\cap C_j$
is acceptable for agents of type $i_1$ (Case 6 is excluded). Thus, Lemma
\ref{every_second} applies to splits $\cC$ and $\cA$ and implies that
every pair of consecutive sets of the mms-split $C_1,\ldots,C_n$ contains
at least one set acceptable to agents of type $i_1$.

We can assume without loss of generality that $C_{1}\subseteq B_{1}$. It
follows that the sets $B_n$ and $C_1$ are jumps.
If the set $C_{2n_3}$ is a jump, then the %cyclic
sequence
\begin{equation*}%\label{sequence3}
\cZ_1=C_{1},\ldots,C_{2n_3},B_{2n_3+1},\ldots,B_n
\end{equation*}
is $\cB\cC$-useful. Similarly, if the set $C_{2n_3-1}$ is a jump, then the
%cyclic
sequence
\begin{equation*}%\label{sequence4}
\cZ_2=C_{1},\ldots,C_{2n_3-1},B_{2n_3},\ldots,B_{n}
\end{equation*}
is $\cB\cC$-useful.

Assume now that the sets $C_{2n_3-1}$ and $C_{2n_3}$ are not jumps. Let $k$,
$1\leq k<2n_3-1$, be the largest index such that $C_k$ is a jump (since $C_1$
is a jump, $k$ is well defined), and let $\ell$ be the smallest integer $j$
such that $2n_3 < j \leq n$ and $C_{j}$ is a jump, if such $j$ exists, or
$\ell=n$ otherwise.

We observe that the sequence
\begin{equation*}%\label{sequence5}
\cZ_3=C_{1},\ldots,C_{k},B_{k+1},\ldots,B_{\ell-(2n_3-k)},C_{\ell-(2n_3-k)+1},\ldots,C_{\ell},B_{\ell+1}\ldots,B_{n}
\end{equation*}
is $\cB\cC$-useful. In particular the set $B_{\ell-(2n_3-k)}$ is a jump.
Indeed, $k<\ell-(2n_3-k)<\ell$ and for $k<j<\ell$ the sets $C_j$ are not
jumps. By Lemma \ref{lem0}, the sets $B_j$ are jumps.

It can be shown similarly that the sequence
\begin{equation*}%\label{sequence6}
\cZ_4=C_{1},\ldots,C_{k},B_{k+1},\ldots,B_{\ell-(2n_3-k)+1},C_{\ell-(2n_3-k)+2},\ldots,C_{\ell},B_{\ell+1}\ldots,B_{n}
\end{equation*}
is $\cB\cC$-useful.

We now observe that the sequence $\cZ_1$ is a special case of the sequence
$\cZ_3$ for $k=2n_3$, and the sequence $\cZ_2$ is a special case of
the sequence $\cZ_4$ for $k=2n_3-1$. Thus, to complete this case, it is
enough to describe allocations based on sequences $\cZ_3$ and $\cZ_4$.

For $k$ even, we construct an allocation based on the sequence
$\cZ_3$ which, we recall is $\cB\cC$-useful. Since no set $A_i\cap C_j$
is acceptable for agents of type $i_1$, by Lemma \ref{every_second} half of
the sets $C_{1},\ldots,C_{k}$ and half of the sets $C_{\ell-(2n_3-k)+1},
\ldots,C_{\ell}$  (the number of sets in each of these intervals is even),
$n_3$ sets in total, are acceptable for agents of type $i_1$. We allocate
these sets to agents of type $i_1$. Agents of type $i_3$ receive the
remaining $n_3$ sets from these intervals. Since at least every other set
of the intervals $B_{k+1},\ldots,B_{\ell-(2n_3-k)}$ and $B_{\ell+1}\ldots,B_n$
in $\cZ_5$ is acceptable for agents of type $i_1$ (by Lemma \ref{every_second}),
these two intervals contain
at least $\left\lfloor\frac{\ell-2n_3}{2}\right\rfloor+
\left\lfloor\frac{n-\ell}{2}\right\rfloor=
\left\lfloor\frac{\ell}{2}\right\rfloor+
\left\lfloor\frac{n-\ell}{2}\right\rfloor-n_3
\geq \left\lfloor\frac{n}{2}\right\rfloor-1-n_3 \geq n_1-n_3$ sets (because Case 5
is excluded, $\left\lfloor\frac{n}{2}\right\rfloor \geq n_1+1$) that are acceptable to agents of type $i_1$. We allocate $n_1-n_3$
of these sets to the agents of type $i_1$ that have not been allocated any set
in the previous stage. The remaining sets, there are $n_2$ of them, are
allocated to agents of type $i_2$. Clearly, this allocation gives rise to a
$\frac{3}{4}$-sufficient allocation for $C$.

For $k$ odd, we construct an allocation based on the $\cB\cC$-useful
sequence $\cZ_4$. By the same argument as above,
at least $\lfloor \frac{k}{2}\rfloor$ of the sets $C_{1},\ldots,C_{k}$ and at
least $\lfloor\frac{2n_3-k-1}{2}\rfloor$ of the sets $C_{\ell-(2n_3-k)+2},
\ldots,C_{\ell}$ are acceptable for agents of type $i_1$. Since $k$ is odd,
the total number of sets in the two sequences that are acceptable to agents
of type $i_1$ is at least $\frac{k-1}{2}+\frac{2n_3-k-1}{2}=n_3-1$. We select
some $n_3-1$ of these sets and allocate them to $n_3-1$ agents of type $i_1$.
The number of sets remaining in the two sequences is $n_3$. We allocate them
to agents of type $i_3$. Moreover, the intervals $B_{k+1},\ldots,
B_{\ell-(2n_3-k)+1}$ and $B_{\ell+1}\ldots,B_{n}$ contain at least
$\left\lfloor\frac{\ell-2n_3+1}{2}\right\rfloor+
\left\lfloor\frac{n-\ell}{2}\right\rfloor=
\left\lfloor\frac{\ell+1}{2}\right\rfloor+
\left\lfloor\frac{n-\ell}{2}\right\rfloor-n_3
\geq \left\lfloor\frac{n}{2}\right\rfloor-n_3 \geq n_1-n_3+1$ sets
that are acceptable to agents of type $i_1$ (as before, we use the fact that
$\left\lfloor \frac{n}{2}\right\rfloor \geq n_1+1$). We allocate some
$n_1-n_3+1$ of these sets to agents of type $i_1$ that have not been allocated
a set before. Finally, we allocate
the remaining sets in these two sequences, there are $n_2$ of them, to agents
of type $i_2$. This allocation gives rise to a $\frac{3}{4}$-sufficient
allocation for $C$.

From now on, we will assume that no $C$-set is included in any $B$-set.
Reasoning as in Case 3, we can prove that also no $B$-set is included in any
$C$-set. Thus, in the remaining part of the proof we may assume that the splits
$\cB$ and $\cC$ are proper relative to each other.

\smallskip
\noindent
{\bf Case 8.} Some set $B_i\cap C_j$ is acceptable to agents of type
$i_3$.

\smallskip
\noindent
Since the splits $\cB$ and $\cC$ are proper relative to each other, we may
relabel the sets if necessary so that for all $i$, $B_i\subseteq C_i\cup
C_{i+1}$, and that $B_1\cap C_1$ or $B_n\cap C_1$ is acceptable for agents
of type $i_3$. We will consider the first case, that is, that $B_1\cap C_1$
is acceptable to agents of type $i_3$. As in several places before, the other
case can be handled in a similar way.

Suppose first that $n_1>n_2$. Since the conditions formulated in Cases 1 and
6 are not satisfied, by Lemma \ref{every_second} every other set in mms-splits
$\cB$ and $\cC$ is acceptable to agents of type $i_1$ (recall that $\cB$ and
$\cA$ are proper relative to each other and so are $\cC$ and $\cA$).
Thus, there are at least $n_3-1$ sets in the interval $C_2,C_{3},\ldots,
C_{2n_3-1}$ that are acceptable to agents of type $i_1$. We allocate these
sets to those agents. We allocate the remaining $n_3-1$ the sets in this
sequence and the set $B_1\cap C_1$ to agents of type $i_3$. Similarly, we
note that the sequence $B_{2n_3},B_{2n_3+1},\ldots, B_n$ contains at least
$\left\lfloor\frac{n-2n_3+1}{2} \right\rfloor\geq \left\lfloor\frac{n}{2}
\right\rfloor-n_3\geq n_1-n_3+1$ sets that are acceptable to agents of type
$i_1$. We allocate some $n_1-n_3+1$ of
such sets to those agents of type $i_1$ that have not yet been allocated a
set. The remaining sets in the sequence, there are $n_2$ of them, are allocated
to agents of type $i_2$. All sets in the assignment are pairwise disjoint
as $B_{2n_3}\subseteq C_{2n_3}\cup C_{2n_3+1}$. Thus, the assignment we
defined gives rise to a $\frac{3}{4}$-sufficient allocation for $C$.

Thus, assume that $n_1=n_2$. In this case agents of type $i_3$ receive
$B_1\cap C_1$ and the sets $C_2,C_{3},\ldots,C_{n_3}$. Agents of types
$i_1$ and $i_2$ distribute among themselves the sets $B_{n_3+1}, B_{n_3+2},
\ldots, B_n$. This is possible as the number of those sets is even, in fact,
equal to $2n_1$ ($=2n_2$), and every other set of this sequence is acceptable
for agents of type $i_1$. Since $B_{n_3+1}\subseteq C_{n_3+1}\cup C_{n_3+2}$,
these sets are disjoint with the sets assigned to agents of type $i_3$ and,
clearly, they are pairwise disjoint themselves. It follows that this allocation gives rise to a $\frac{3}{4}$-sufficient allocation for $C$.

\smallskip
\noindent
{\bf Case 9.} None of the conditions formulated in Cases 1-8 holds.

\smallskip
\noindent
We may assume that any pair of splits $\cA$, $\cB$, $\cC$ are proper
relative to each other. Moreover, relabeling the sets if necessary,
we may also assume that $A_i\subseteq C_i\cup C_{i+1}$ and $B_i\subseteq
C_i\cup C_{i+1}$. Thus, $A_i$ and $B_i$ overlap and, consequently, $A_i
\subseteq B_i\cup B_{i+1}$ or $A_i\subseteq B_{i-1}\cup B_{i}$. Both cases
can be dealt with similarly, so we assume that $A_i\subseteq B_i\cup B_{i+1}$.

We note that for every $i$
\[
C_{i}\cup C_{i+1}\cup C_{i+2}=(A_{i-1}\cap C_{i})\cup A_i\cup (A_{i+1}\cap B_{i+1})\cup (B_{i+2}\cap C_{i+2}).
\]
The value of $C_{i}\cup C_{i+1}\cup C_{i+2}$ for agents of type $i_3$ is $3$.
The conditions formulated in Cases 1, 4 and 8 do not hold, so the sets
$A_{i-1}\cap C_{i}$, $A_{i+1}\cap B_{i+1}$ and $B_{i+2}\cap C_{i+2}$ are not
acceptable for agents of type $3$. Thus, the value of $A_i$ for agents of
type $i_3$ is larger than $3-3\cdot\frac{3}{4}=\frac{3}{4}$. As $i$ was
arbitrary, it follows that every $A$-set is acceptable for agents of type
$i_3$.

We have $n_2\leq n_1<\left\lfloor\frac{n}{2}\right\rfloor$, where the latter
inequality holds because of Case 5 being excluded. Since at least one in every
two consecutive sets $A_i, A_{i+1}$ is acceptable to agents of type $i_2$, we
assign $n_2$ of such sets to agents of type $i_2$. All remaining $A$-sets are
acceptable for agents of types $i_1$ and $i_3$. Thus, they can be allocated
among them arbitrarily yielding a $\frac{3}{4}$-satisfying allocation.
\end{proof}

\end{document}